\documentclass[letterpaper,12pt]{article}
\usepackage{biblatex} 
\AtBeginBibliography{%
  \setlength{\emergencystretch}{3em}%
}
\usepackage{graphicx} 
\usepackage{palatino,mathrsfs}
\usepackage[mathcal]{euler}
\usepackage{xcolor}
\usepackage{epstopdf,epsfig}
\usepackage{pdfpages}
\usepackage{hyperref}
\usepackage{comment} 
\usepackage{amsmath,amsthm,amsfonts,amssymb,latexsym,amscd,enumerate,url,hyperref}
\usepackage{amssymb}
\usepackage{url}
\usepackage{mathtools}
\usepackage{graphicx}
\usepackage[all,knot]{xy}
\xyoption{arc}
\usepackage{multirow}
\usepackage{caption}
\usepackage[new]{old-arrows}

\newtheorem*{theorem*}{Theorem}
  
 \newtheorem*{corollary*}{Corollary}

\newtheorem*{lemma*}{Lemma}

\newtheorem*{proposition*}{Proposition}

\newtheorem*{claim*}{Claim}
\theoremstyle{definition}

\newtheorem*{defn*}{Definition}
 
\newtheorem*{example*}{Example}
\newtheorem*{examples*}{Examples}

\newtheorem*{question*}{Question}
  
\newtheorem*{problem*}{Problem}
  
\newtheorem*{remark*}{Remark}
\newtheorem*{fact*}{Fact}

\def\R{\mathbb R}
\def\C{\mathbb C}
\def\N{\mathbb N}
\def\Z{\mathbb Z}

\def\fg{\mathfrak{g}}
\def\fk{\mathfrak{k}}
\def\fs{\mathfrak{s}}

\title{ A new model  for the  quantum mechanics of the hydrogen atom}
\author{Joseph Bernstein\thanks{School of Mathematical Sciences, Tel Aviv University, Tel Aviv 6997801, Israel.}  \hspace{1pt} and   Eyal Subag\thanks{Department of Mathematics, Bar-Ilan University, Ramat-Gan, 5290002 Israel.}}
\date{\today}

\begin{document}

\maketitle
\begin{abstract}
To every  Lorentzian quadratic space \((\mathcal{V},q)\) of
even dimension $n$ such that $n\geq 4$
we attach a canonical algebraic quantum mechanical model for a  corresponding  generalized hydrogen atom system.

In our  model
the configuration space is  the regular null cone \(C\) of the quadratic space. The Hilbert space $\boldsymbol{H}$
is a canonical $L^2$ space  on the cone $C$,
and observables are realized in the algebra $\mathcal{D}(C)$ of  algebraic differential operators  on $C$. 
We also construct a distinguished Schwartz space $\mathcal{S}(\boldsymbol{H})\subset \boldsymbol{H}$, 
which carries a self-adjoint action of $\mathcal{D}(C)$ and encodes the boundary 
conditions of the standard theory. The role of the Schr\"odinger operator is played by a one-parameter 
Schr\"odinger family of operators in $\mathcal{D}(C)$. We explain how the model relates to the realization of the  minimal representation of $O(n,2)$ on $\boldsymbol{H}$.

For $n=4$, which corresponds to the physical  hydrogen atom system, 
we prove that  the spectrum of the Schr\"odinger family  on the upper-half component \(\mathcal{S}(\boldsymbol{H})_+\) coincides 
with the usual spectrum of the hydrogen atom and that the corresponding solution spaces 
recover the standard physical solutions.

The spectrum of the Schr\"odinger family  on the lower-half component 
\(\mathcal{S}(\boldsymbol{H})_-\) gives additional positive-energy solution 
spaces not present in the usual formulation.
\end{abstract}

\tableofcontents

\section{Introduction}
\subsection{Overview}
The hydrogen atom system is a fundamental well-studied example of a quantum-mechanical system which played a central role in the development of quantum mechanics, see  e.g., \cite{ LandauLifshitz,MR2156403,Wulfman,Wybourne}. The common model for the spinless non-relativistic  system is given by the Hilbert space $L^2(\R^3)$,  together with  the densely defined (normalized) Schr\"{o}dinger operator  
\[{ H_{phys}=-\Delta  -\frac{\kappa}{r},} \] 
where $\kappa$ is a positive constant, $r$ is the radial coordinate, and $ \Delta$ is the Laplacian on $\R^3$.
Solutions of 
the  time-independent Schr\"{o}dinger equation, 
\[H_{phys}\psi =E \psi   \]
determine the time-evolution of the system, where  $E\in \R$  is the energy and   $\psi$ is a function on $\R^3\setminus\{0\}$ 
 satisfying certain boundary conditions.

The spectrum of ${H}_{phys}$, that is the collection of all admissible energies, is
  \[ \operatorname{Spec}({H}_{phys})=\left\{-\frac{\kappa^2}{(2n)^2}|n\in \N\right\}\cup[0,\infty).\]  
  
To deal with degeneracy of the  energy eigenspaces one looks for solutions  of the Schr\"{o}dinger equation that are   simultaneously  eigenfunctions for the  total angular momentum    operator $L^2$ and for its projection on the $z$-axis $L_z$. We will refer to these solutions as \textit{physical solutions}, and denote   the linear span of such solutions for a given $E\in \operatorname{Spec}(H_{phys})$ by $\operatorname{PhysSol}(E)$. 
These solution spaces lead to several layers of symmetries in which the largest symmetry group is $O(4,2)$.    In physics, it is called \textit{the dynamical symmetry}; see, e.g., \cite{sym12081323}.
The  model we just described has several flaws.  
\begin{itemize}
    \item The  Schr\"{o}dinger operator seems ad hoc and definitely not canonical.
    \item The  Schr\"{o}dinger operator has  a singularity at the origin and it involves the non-algebraic square root function via the radial coordinate. 
    \item The Schr\"{o}dinger operator is not directly  a part of the action of the Lie algebra $\mathfrak{o}(4,2)$ that arises from the above-mentioned  $O(4,2)$ symmetry.
    \item The  boundary conditions that are used for the physical solutions do not have a conceptually clear motivation.
\end{itemize}
The purpose of this paper is to present an algebraic model for the quantum-mechanical system of the  hydrogen atom that does not suffer from the above-mentioned flaws and recovers the known spectrum, the physical solutions and their symmetries.

We associate a model with  every Lorentzian quadratic space $(\mathcal{V},q)$ of an even dimension greater than or equal to 4.  To simplify the discussion here we will start with   $\mathcal{V}=\R^4$ equipped with the Lorentzian quadratic form $q(x,y,z,w)=x^2+y^2+z^2-w^2$.  The general case is similar and is being  discussed in the following sections.

We denote the corresponding cone by 
\[C(\R)=\{ v\in \R^4| q(v)=0\},  \]
and its regular part $C(\R)\setminus\{0\}$  by  $C_0(\R)$. The regular part of the cone  splits as a union of two connected components, the upper half cone  $C_{+}(\R)$, and the lower half cone $C_{-}(\R)$.

The cone carries a canonical invariant measure $|\omega|$, which  allows us to define the first ingredient of our model which is 
the canonical  Hilbert space,  $\boldsymbol{H}=L^2({C}_0(\R),|\omega|)$  of square integrable functions on the cone.

 The second ingredient is the algebra $\mathcal{D}(C)$ of  algebraic differential operators on  $C(\R)$. 
 
 There is a natural action $\rho^{\infty}$ of the algebra  $\mathcal{D}(C)$  on $C^{\infty}(C_0(\R))$ the space of smooth functions on $C_0(\R)$. We construct  a distinguished subspace $\mathcal{S}(\boldsymbol{H})$ of $\boldsymbol{H}$  that we call the  Schwartz space of $\boldsymbol{H}$. The subspace $\mathcal{S}(\boldsymbol{H})$  is contained in $C^{\infty}(C_0(\R))$ and invariant under the action of the algebra $\mathcal{D}(C)$. Moreover this action $\rho$ is self-adjoint. The self-adjoint action $(\rho,\mathcal{S}(\boldsymbol{H}))$ of the algebra $\mathcal{D}(C)$ is the third ingredient of our model.
 We remark that the 
Schwartz subspace $\mathcal{S}(\boldsymbol{H})$ encodes the fact that $\boldsymbol{H}$ carries a distinguished unitary representation of an orthogonal group $G(\mathcal{V})$ (that extends $O(\mathcal{V})$ and is isomorphic to $O(4,2)$).

 The last ingredient of our model  is the \textit{Schr\"{o}dinger family}, which is a one-parameter family of differential operators on the cone $E\longmapsto S(E)\in  \mathcal{D}(C)$, parameterized by $E\in \C$.  To construct this family  we fix a positive $\kappa$ and we further assume that we are given a linear functional $w\in \mathcal{V}^*$ of norm $-1$. It is easy to pass from one choice of $\kappa$ and $w$ to another, and in this sense the actual choice is not important. Explicitly,
 \[ S(E)=
 H_w+\kappa+Ew ,\] 
 where $H_w$ is a certain second order differential operator on $C$. 
The  Schr\"o\-dinger family corresponds to the  Schr\"odinger operator $H_{phys}$ in the standard  model of the hydrogen atom.

We define  the spectrum of the Schr\"{o}dinger family in $\mathcal{S}(\boldsymbol{H})$ to be the collection of all $E\in \C$ for which the kernel of the Hermitian adjoint
\[
S(E)^{\dagger}:(\mathcal{S}(\boldsymbol{H}))^{\dagger}\longrightarrow (\mathcal{S}(\boldsymbol{H}))^{\dagger}
\]
is non-zero, where $(\mathcal{S}(\boldsymbol{H}))^{\dagger}$ denotes the  Hermitian dual of $\mathcal{S}(\boldsymbol{H})$.

The Schwartz space $\mathcal{S}(\boldsymbol{H})$ as a representation of 
$\mathcal{D}(C)$ has a canonical decomposition into   invariant subspaces  $\mathcal{S}(\boldsymbol{H})=\mathcal{S}(\boldsymbol{H})_{+}\oplus \mathcal{S}(\boldsymbol{H})_{-}$, where $\mathcal{S}(\boldsymbol{H})_{\pm}$ consists of those functions in $\mathcal{S}(\boldsymbol{H})$ that are supported in the half cone   $C_{\pm}(\R)$.  
We  prove that the spectrum of the  Schr\"{o}dinger family in the upper half cone $\mathcal{S}(\boldsymbol{H})_{+}$   is equal to the spectrum of the Schr\"{o}dinger operator  $H_{phys}$. Moreover the physical solution spaces can be identified  with solution spaces (kernels) of the Schr\"{o}dinger family in $\mathcal{S}(\boldsymbol{H})_{+}$. 
We also show  that the spectrum of the Schr\"{o}dinger family in the lower half cone $\mathcal{S}(\boldsymbol{H})_{-}$   is equal to  $(0,\infty)$.   These    solutions  are not seen in the standard description of the system in physics. It is natural to ask for a possible physical  meaning of these new solutions. We do not know the answer to this question.

Our model is completely algebraic and does not suffer from the above-mentioned flaws of the classical model. 
In addition,
our model gives a conceptual explanation of the formula for the $n$th
negative-energy level: the factor $2n$ arises as an $SO(2)$-weight in a
 discrete-series representation of $SL_2(\R)$; see
Subsection~\ref{623}, especially the final remark. This formula yields the
Rydberg formula for the spectral lines of hydrogenic systems.

Before giving a detailed account of our main results we will make two comments that put this work within a larger context and have  mathematical implications that go beyond this work. 

The  boundary conditions that are used in the classical model do not arise in our model and they are implicitly hiding in the Schwartz  space $\mathcal{S}(\boldsymbol{H})$. This  property of implicit boundary conditions through the canonical smooth subspace may have uses elsewhere. For example, when solving differential equations within a unitary representation of a reductive Lie group.

 Another point that we   discuss below is a group-theoretical and repre\-sentation-theoretical description of the spectrum problem for the hydrogen atom. We show how the Schr\"{o}dinger family corresponds to a one-parameter family of elements in the universal enveloping algebra of $SL_2(\R)$. Moreover we show that the calculation of the spectrum of the Schr\"{o}dinger family in $\boldsymbol{H}$ can be reduced to calculations completely within discrete series representations of $SL_2(\R)$.

 We finish this part of the introduction by describing the main observation that suggested to us the existence of our new model. We noticed two a priori unrelated situations in which there is an expected symmetry of $SO(n)$ but in practice the symmetry is of a larger orthogonal group. 
 
 The first instance is the usual model of the hydrogen atom system, in which there is a clear symmetry of $SO(3)$ and a larger hidden symmetry\footnote{In fact there are different hidden symmetry groups within the dynamical symmetry that are fibers of a nice family, see \cite{MR4460278,MR3827131}.}  of $O(4,2)$ that is called the dynamical symmetry.

 The other instance is of the algebra of differential operators $\mathcal{D}(C)$ of the complex algebraic variety $C$ that is determined by a quadratic cone in $\C^4$. 
It is easy to see that first order elements in $\mathcal{D}(C)$ naturally contain $\mathfrak{so}(4,\C)$. Much less obvious is the fact that  the subspace of second order elements in $\mathcal{D}(C)$  contains a copy of $\mathfrak{so}(6,\C)$. This was the starting point in our search for a new model for the hydrogen atom system which eventually led to the results presented in this paper.

 Recently we discovered works by Meng  \cite{Meng11,Meng13} that introduce a model close to our model. 
Meng considers a more general case that includes the case of the Lorentz group acting on a cone. 

  Meng starts from a simple Euclidean Jordan algebra $J$ and builds a Lie algebra $\mathfrak{co}$, the conformal Lie algebra of $J$, using the
  Kantor–Koecher–Tits
  construction.

 Then Meng constructs a Kepler cone $\mathcal{C}$ associated to $J$ and considers the Hilbert 
space $\mathcal{H}= L^2(\mathcal{C})$. He observes that  the Lie algebra $\mathfrak{co}$ acts on functions on $\mathcal{C}$ via differential operators. 

Then he defines a Schr\"{o}dinger operator  in terms of the Lie algebra $\mathfrak{co}$ and       computes its discrete spectrum.

In the case of a Jordan algebra $J$ corresponding to a   Lorentzian  space the Kepler cone $\mathcal{C}$
constructed by Meng coincides with our cone $C_+(\R)$, the Hilbert space $\mathcal{H}= L^2(\mathcal{C})$ is our space $\boldsymbol{H}_+=L^2(C_+(\R))$ and his Lie algebra $\mathfrak{co}$ is isomorphic to (a real form of) our  quantum dynamical symmetry $\fs$ (see subsection \ref{222}).

Our approach is much more algebraic. We mostly work with the algebra $\mathcal{D}(C)$ of algebraic 
differential operators on $C$; the Lie algebra $\fs$ arises as a convenient generating subspace of $\mathcal{D}(C)$ and our Schr\"{o}dinger family lives inside $\mathcal{D}(C)$. 

Also, one of the main   ingredients of our model is the Schwartz space $\mathcal{S}(\boldsymbol{H}) \subset \boldsymbol{H}$ and a self-adjoint action of the algebra $\mathcal{D}(C)$ on the Schwartz space.

 The Schwartz space implicitly includes boundary conditions.   This allows us to investigate the continuous spectrum in addition to the discrete spectrum.
 
 \subsection{Structure of the paper  and main results}\label{1.2} 
We review the main results of the paper following the chronological order in which they appear in the rest of the paper.

\subsubsection{Lorentzian quadratic space}
A  Lorentzian quadratic space is an $n$-dimensional real  vector space $\mathcal{V}$ along with a Lorentzian quadratic form $q$ that  is  of signature $(n-1,1)$. We will mostly deal with a Lorentzian quadratic space together with a fixed 
linear functional  $w$ such that $q^*(w)=-1$, where $q^*$ is the dual form on the dual vector space. We will refer to this data as a pointed  Lorentzian quadratic space and denote this in short by PLQS. With every $n$-dimensional PLQS we  associate a corresponding model for a  quantum-mechanical system of an  $(n-1)$-dimensional generalized hydrogen atom system. The model described in the introduction above is the  one  associated with the PLQS with $\mathcal{V}=\R^4$, $q$ is the standard Lorentzian quadratic form, represented by the diagonal matrix $\operatorname{diag}(1,1,1,-1)$, and $w$ being the linear functional corresponding to the last coordinate.

Throughout this part of the introduction (Section \ref{1.2})
 we assume that $(\mathcal{V},q,w)$ is a four-dimensional PLQS.
This case corresponds to the case of the hydrogen atom system. We denote the complexification of $(\mathcal{V},q)$  by $(V,Q)$. The corresponding cone $C$ is
 the complex algebraic variety  $\{v\in {V}| Q(v)=0 \}$.

\subsubsection{The algebra of differential operators on the cone}
In Section \ref{se2} we study 
 the  algebra $\mathcal{D}(C)$  of   differential operators on  $C$. Later on we study its action on smooth functions on the corresponding real cone. 

The algebra $\mathcal{D}(C)$ is filtered by degree of  differential operators and each piece in the 
filtration has a grading according to homogeneity. This gives a decomposition   
 \[\mathcal{D}(C)= \bigcup_{k\in \N_0}\left(\mathcal{D}^{k}(C)\right)= \bigcup_{k\in \N_0}\left(\oplus_{\ell\in \Z}\mathcal{D}^{k,\ell}(C)\right), \]
 where $\N_0=\Z_{\geq0}$.
 It is  known that the algebra $\mathcal{D}(C)$ is not generated by differential operators of degree one but it is generated by differential operators of degree two. Namely let  
 \[\tilde{\fs}:=\oplus_{\ell+k=1, k\leq 2}\mathcal{D}^{k,\ell}(C)=\mathcal{D}^{0,1}(C)\oplus \mathcal{D}^{1,0}(C) \oplus \mathcal{D}^{2,-1}(C),\]  it is clear that  $\tilde{\fs}$ is a Lie-algebra . The derived Lie-algebra  
 $\fs:=[\tilde{\fs},\tilde{\fs}]$, is   isomorphic to $\mathfrak{o}(6,\C)$ and moreover 
 generates $\mathcal{D}(C)$.

 The Lorentzian real form 
 $(\mathcal{V},q)$  of $(V,Q)$
  leads to a real form $\fs(\R)$ of $\fs$. 
 In section \ref{s2.5} we show that $\fs(\R)$, the algebra of differential operators on the cone with real coefficients is isomorphic to  $\mathfrak{o}(4,2)$.

\subsubsection{A special endomorphism of the vector space $\mathcal{D}(C)$}
The form $Q$ defines a canonical morphism $\Phi^Q$ of the vector space $\mathcal{D}(C)$. Explicitly,  if in coordinates   $Q$  is given by  $\sum_{i,j=1}^4B_{ij}z_iz_j $, then 
for any $d\in \mathcal{D}(C)$,
\[\Phi^Q(d)= \sum_{i,j=1}^4B_{ij}z_idz_j.\]

This is  discussed in Section \ref{s224}. 
We  describe the restriction of $\Phi^Q$ to the subspace $\mathcal{D}^{2,-1}(C)$ of $\fs$. We use the following result. 
  
\begin{proposition*}
The morphism $\Phi^Q$   defines an  isomorphism    \[\Phi^Q|_{\mathcal{D}^{2,-1}(C)}: \mathcal{D}^{2,-1}(C)\xrightarrow{\sim}\mathcal{D}^{0,1}(C)=V^*.\]
\end{proposition*}

In fact we implicitly encountered the isomorphism $\Phi^Q|_{\mathcal{D}^{2,-1}(C)}$  before. Namely, the  second order differential operator $H_w$ that appeared in the above-mentioned  Schr\"{o}dinger family $S(E)$  is defined using this isomorphism. More on this point can be found in Sections \ref{118} and \ref {se5} below.

In Appendix \ref{Ap3} we give a group-theoretical characterization of  $\Phi^Q|_{\mathcal{D}^{2,-1}(C)}$.  
\subsubsection{The dual pair symmetry}
In Section \ref{PLQS} we show that $\fs(\R)$ has a canonical dual pair, isomorphic to 
$(\mathfrak{so}(3),\mathfrak{sl}_2(\R))$  in $\mathfrak{so}(4,2)$. This explains how a PLQS gives rise to a dual pair symmetry in $\mathcal{D}(C)$.

\subsubsection{The Hilbert space}
In section \ref{sub4.2} we describe 
a canonical invariant measure $|\omega|$ on $C(\R)$. Using the measure we  attach to our PLQS a canonical Hilbert space $\boldsymbol{H}=L^2({C}_0(\R),|\omega|)$,  consisting of square integrable functions on  $C_0(\R)$. 

\subsubsection{Dense actions}
In Section \ref{sub4.1}   
 we recall  the  notion of a partial  action of a $*$-algebra $A$ on a Hilbert space $\boldsymbol{H}$. A partial action  of $A$ on $\boldsymbol{H}$ is given by an  action $\rho:A\longrightarrow \operatorname{End}(\mathcal{W})$ of the $*$-algebra $A$,
on a linear subspace $\mathcal{W}$ of $\boldsymbol{H}$. The subspace $\mathcal{W}$ is called the domain of the action.

When $\mathcal{W}$ is dense the partial action is called a dense action. 
For a dense action there exists  a well-defined adjoint partial action
$\rho^{*}:A\longrightarrow \operatorname{End}(\mathcal{W}^{*})$, where $\mathcal{W}^{*}\subset \boldsymbol{H}$ is the adjoint domain.

We  describe related notions of symmetric, self-adjoint and essentially self-adjoint dense actions of  a $*$-algebra on a Hilbert space.  

Our treatment is a    natural generalization of the theory of  unbounded operators on Hilbert spaces.

\subsubsection{The self-adjoint action on the Schwartz space}
There is a natural action $\rho^{\infty}$ of the  algebra $\mathcal{D}(C)$  on the space $C^{\infty}(C_0(\R))$ of smooth functions on $C_0(\R)$.

In Section \ref{sub4.3} we show that $\rho^{\infty}$ restricts to a symmetric action on the domain $C^{\infty}_c(C_0(\R))$  of smooth compactly supported functions on $C_0(\R)=C(\R)\setminus\{0\}$.

Our main goal in  Section \ref{ddd} is to obtain  a distinguished Schwartz subspace $\mathcal{S}(\boldsymbol{H})$ of $\boldsymbol{H}$ that carries a self-adjoint action of the algebra $\mathcal{D}(C)$ extending the action on $C^{\infty}_c(C_0(\R))$. In order to do that, in Section \ref{seext} 
we extend the group $O(q)\simeq O(3,1)$ to a group $G(\mathcal{V})$ isomorphic to $O(4,2)$. 
Then we  show that the natural unitary representation of  the group $O(q)$  on the canonical Hilbert space $\boldsymbol{H}$ can be   extended to a distinguished  unitary irreducible representation $\pi$ of $G(\mathcal{V})$ on the   Hilbert space $\boldsymbol{H}.$

We define the Schwartz space $\mathcal{S}(\boldsymbol{H})$ to  be the canonical subspace of $\boldsymbol{H}$ consisting of smooth vectors  for $\pi$.
Using the inclusions 
$$  C^{\infty}_c(C_0(\R)) \subseteq \mathcal{S}(\boldsymbol{H}) \subseteq C^{\infty}(C_0(\R)) $$    and the fact that 
$\mathcal{S}(\boldsymbol{H})$ is stable under the action  $\rho^{\infty}$
of the algebra
$\mathcal{D}(C)$
we observe that
$\mathcal{S}(\boldsymbol{H})$  is a self-adjoint extension of the symmetric   action on $C^{\infty}_c(C_0(\R))$.

\subsubsection{The Schr\"{o}dinger family}\label{118}

In Section~\ref{se5}, with every PLQS $(\mathcal{V},q,w)$ such that
$q^*(w)=-1$, and every $\kappa\in\R$, we associate a one-parameter
family of elements of $\widetilde{\fs}\subset\mathcal{D}(C)$, the
\textit{Schr\"odinger family}, via
\[
S(q,w,\kappa;\lambda):=\Psi^Q(w)+\kappa+\lambda w
\in\widetilde{\fs},\qquad \lambda\in\C,
\]
where  $\Psi^Q$ is a   multiple of the inverse of $\Phi^Q|_{\mathcal{D}^{2,-1}(C)}: \mathcal{D}^{2,-1}(C)\xrightarrow{\sim}V^*$.

The parameter $\kappa$ plays the role of the strength of the Coulomb potential. 
For every $\lambda\in \C$, the operator   $S(q,{w},\kappa;\lambda)$ belongs to the universal enveloping algebra of a canonical subalgebra  $\mathfrak{l}$ of $\fs$, where $\mathfrak{l}\simeq \mathfrak{sl}_2(\C)$. 
This gives rise to    a one-parameter family in the universal enveloping algebra of $\mathfrak{sl}_2(\C)$ (that we call  \textit{the abstract Schr\"odinger family}). 
\subsubsection{The model}
Altogether, starting with a four-dimensional PLQS and $\kappa\in \R$, the corresponding  model for the quantum-mechanical system of the hydrogen atom consists of 
\begin{enumerate}
    \item The canonical Hilbert space, $\boldsymbol{H}=L^2({C}_{0}(\R)),$
\item The $*$-algebra of differential operators on the cone, $\mathcal{D}(C)$,
\item The distinguished self-adjoint action $(\rho,\mathcal{S}(\boldsymbol{H}))$, of the algebra  $\mathcal{D}(C)$ on $\boldsymbol{H}$,
\item The Schr\"{o}dinger family 
\[S(q,w,\kappa;\lambda)=  \Psi^Q\left({w} \right) + \kappa +\lambda {w}  \in \mathcal{D}(C).  \]
\end{enumerate}

\subsubsection{The spectrum}
In section \ref{se6} we define the spectrum of a family of operators on a topological vector space using the  Hermitian dual.

Under the action of the connected component  $G(\mathcal{V})_0\simeq SO_0(4,2)$ the   Hilbert space $\boldsymbol{H}$ has  a canonical decomposition $\boldsymbol{H}_+\oplus \boldsymbol{H}_-$ into irreducible sub-representations. The  subspace $\boldsymbol{H}_{+}$ consists of functions in $\boldsymbol{H}$ that are supported in the upper half cone and similarly $\boldsymbol{H}_{-}$  consists of functions that are supported in the lower half cone. We set $\mathcal{S}(\boldsymbol{H})_{\pm}:=\mathcal{S}(\boldsymbol{H})\cap \boldsymbol{H}_{\pm}$.
We prove that the spectrum of the Schr\"{o}dinger family in $\mathcal{S}(\boldsymbol{H})_+$ is equal to the spectrum of the Schr\"{o}dinger operator in physics. Below is the exact statement.  
\begin{theorem*}
Let $(\mathcal{V},q,w)$ be a $4$-dimensional  PLQS and  let $\boldsymbol{H}$ be the corresponding canonical Hilbert space. 
For any $\kappa >0$,  
  \[\operatorname{Spec}_{\mathcal{S}(\boldsymbol{H})_{\pm}}\{ S(q,w,\kappa;\lambda)
 |\lambda\in \C \}= \begin{cases}
    \left\{-\frac{\kappa^2}{(2n)^2}
|n\in \N\right\}\cup[0,\infty), & +\\
(0,\infty), &-.
  \end{cases}\] 
\end{theorem*}

In particular  the spectrum of the Schr\"{o}dinger family in $\mathcal{S}(\boldsymbol{H})_+$  is equal to  $\operatorname{Spec}(H_{phys})$, the spectrum in the standard model in physics. We identify the Hermitian-dual solution spaces for the Schr\"{o}dinger family in $\mathcal{S}(\boldsymbol{H})_+$ with the physical solutions. The spectrum in the full Schwartz  space $\mathcal{S}(\boldsymbol{H})$ has more solutions for $\lambda > 0$ than in the classical model in physics.

\subsubsection{Concluding remarks}
Our suggested model does not suffer from the flaws of the classical model. It is defined with respect to a four-dimensional PLQS and $\kappa\in \R$. Up to a natural  isomorphism, the model does not depend on the PLQS and in that sense it is canonical. 

  The main differences with the standard model are
\begin{enumerate}
\item We use the cone $C$ instead of  $\R^3$ as our configuration space.  As a result, the   group of geometric symmetries of our configuration space is $O(q)\simeq O(3,1)$ rather than  $O(3)\ltimes \R^3$ for $\R^3$. 
\item We only use algebraic operators with no singularities and in particular the Schr\"{o}dinger family belongs to $\tilde{\fs}\subset \mathcal{D}(C)$ and  leads to a reductive dual pair symmetry of the type of $(\mathfrak{so}(3),\mathfrak{sl}_2(\R))$  in $\mathfrak{so}(4,2)$. 
\item We do not use any specific boundary conditions for solutions of our equations. The distinguished Schwartz  space $\mathcal{S}(\boldsymbol{H})$ encodes them in an indirect way. 
\end{enumerate}

 We compute the spectrum of the Schr\"{o}dinger family  in the Schwartz space $\mathcal{S}(\boldsymbol{H})$  and show that  it coincides with the spectrum in physics and solutions in $\mathcal{S}(\boldsymbol{H})_+$  correspond to the  usual solutions in physics.

This research was supported by the Israel Science Foundation   (grant No. 1040/22).

\section{The algebra of differential operators on a cone}\label{se2} 
Let $(V,Q)$ be a finite-dimensional complex nondegenerate quadratic space and $C\subset V$ be the associated  cone.
In this section we discuss some known and some lesser-known facts about  the algebra  $\mathcal{D}(C)$ of differential operators on $C$. We will specify and study a canonical Lie-subalgebra $\fs$ of  $\mathcal{D}(C)$ that generates it.   We will also study related real structures of $\mathcal{D}(C)$.
\subsection{Differential operators}\label{s2.1}
\subsubsection{Differential operators on linear spaces}\label{Sec2.1}
In this work we will consider complex (and occasionally also real)  affine algebraic varieties, which for our purposes can be given as the solution  set of a system of   polynomial equations in a complex affine space. For such a variety $X\subset \C^n$, its ring of polynomial functions $\mathcal{O}_X$ is a quotient of $\C[z_1,z_2,\ldots,z_n]$.  Grothendieck defined a differential operator of 
degree less than or equal to
$k\in\N_0$ on $X$ to be a linear operator  $D\in \operatorname{End}_{\C}(\mathcal{O}_X)$,  for which 
\[ [{f_0},[{f_1},...,[{f_k},D]]...]=0,\]
for any $f_0,f_1,...,f_k\in \mathcal{O}_X$, and where we treat 
$\mathcal{O}_X$ as a subalgebra of $\operatorname{End}_{\C}(\mathcal{O}_X)$. 
 The collection $\mathcal{D}^k(X)$ of all differential operators of degree less than or equal to  $k$ on $X$ defines an increasing exhaustive filtration of $\mathcal{D}(X)$, the algebra of differential operators on $X$. 

When $X$ is smooth,  $\mathcal{D}(X)$ is generated as an  algebra   by $\mathcal{O}_X$ and $\operatorname{Der}(\mathcal{O}_X)=\{D\in \mathcal{D}^1(X)| D(1)=0\}$. In general, beyond the smooth case, $\mathcal{D}(X)$ may not even be Noetherian, e.g., see the case of the  cubic  cone \cite{cubiccone}. 

Consider the case in which  
  the algebraic variety  $X$ is a finite-dimensional complex vector space $V$.
  In this case we have an embedding $V^*  \hookrightarrow  \mathcal{O}_X  \hookrightarrow \mathcal{D}^0(X) $ sending $\varphi$ to multiplication by $\varphi$ as an element in $ \mathcal{D}^0(V)$. 
  In addition, we have an embedding of $V$ in $\mathcal{D}^1(V)$ sending a vector $v$ to the corresponding directional derivative $\partial_v$. 
  We note that for any $\varphi \in V^*$ and  $v\in V$, \[[\partial_v,{\varphi}]=\partial_v(\varphi).\]
The  algebra $\mathcal{D}(V)$ of differential operators  on $V$  is generated by the images of $V$ and $V^*$ in  $\mathcal{D}(V)$. This algebra is called the Weyl algebra of $V$.
 \subsubsection{Differential operators on cones}\label{s212} 
From now on we will assume that $(V,Q)$ is a complex nondegenerate quadratic space with  $\operatorname{dim}(V)=n>2$. We denote the  associated symmetric bilinear form by $B_Q$ or by $B$, and the corresponding cone by 
\[C=C_Q=\{v\in V| Q(v)=0 \}. \]
 In contrast to the smooth case, $\mathcal{D}(C)$ is not generated by $\mathcal{D}^1(C)$. However,  $\mathcal{D}(C)$ is generated by $\mathcal{D}^2(C)$, and moreover it is isomorphic to a quotient of the universal  enveloping algebra of $\mathfrak{so}(n+2,\C)$.
 
 We denote by $I=I_Q$ the ideal of $\mathcal{O}_V$ that is generated by $Q\in  \mathcal{O}_V$.  
 
The ring $ \mathcal{O}_C$ of polynomial functions on $C$ is given by  
 $\mathcal{O}_V/I$.

 We say that a differential operator $D\in \mathcal{D}(V)$  is a \textbf{lift} of a differential operator $D'\in \mathcal{D}(C)$ (or equivalently $D'$ is a \textbf{restriction} of $D$) if 
\[D'(f|_{C})= D(f)|_{C}, \quad \forall f\in \mathcal{O}_V. \]
A differential operator  $D\in \mathcal{D}(V)$  has a restriction to a differential operator on $C$ if and only if $D(I)\subset I$.
We define the  subalgebra of    $\mathcal{D}(V)$ consisting of all restrictable operators  by
\[\mathcal{D}(V)_I=\{D\in\mathcal{D}(V)|D(I)\subset I  \}.\]
The following result, which follows from Proposition  1.6 of \cite{SSP}, explains how a differential operator on the cone can be described by differential operators on the ambient vector space. 
\begin{proposition*}\label{pr121}
The natural map $ \mathcal{D}(V)_I\longrightarrow \mathcal{D}(C)$  is onto and has $I\mathcal{D}(V)$ as its kernel.  
\end{proposition*}

From now on we  identify $\mathcal{D}(C)$ with $\mathcal{D}(V)_I/I\mathcal{D}(V)$.

\subsection{Fine structure of the algebra $\mathcal{D}(C)$}\label{Sec2.3}
\subsubsection{The $\Z$-grading}
The group $\C^{\times}$ acts via rescaling  on $\mathcal{O}_V$.
This induces an action  on $\mathcal{D}(V)$
and under this action for each $k\in \N_0$,  $\mathcal{D}^k(V)$ is invariant and hence carries a $\Z$-grading: \[\mathcal{D}^k(V)=\oplus_{\ell\in \Z}\mathcal{D}^{k,\ell}(V), \]
with $\mathcal{D}^{k,\ell}(V)=\left\{d\in \mathcal{D}^{k}(V)| \alpha \cdot d= \alpha^{\ell} d, \forall \alpha \in \C^{\times}  \right\}$. 
Moreover, since $C$ is stable under the action of $\C^{\times}$, the same holds for $\mathcal{D}(C)$:
 \[\mathcal{D}^k(C)=\oplus_{\ell\in \Z}\mathcal{D}^{k,\ell}(C). \] 
 For any $k\in \N_0$ and $\ell\in \Z$, we set  $\mathcal{D}^k(V)_I:=\mathcal{D}(V)_I\cap \mathcal{D}^k(V)$ and $\mathcal{D}^{k,\ell}(V)_I:=\mathcal{D}(V)_I\cap \mathcal{D}^{k,\ell}(V)$.  Since $\mathcal{D}^k(V)_I$ is stable under the algebraic action of $\C^{\times}$ we must have
\[\mathcal{D}^k(V)_I=\oplus_{\ell\in \Z}\mathcal{D}^{k,\ell}(V)_I. \]
\subsubsection{The canonical generating Lie subalgebra}\label{222}
Long ago it was discovered by I.M. Gelfand, S.I. Gelfand, and the first author that $\mathcal{D}(C)$ is generated   by the  Lie algebra   
 \[\tilde{\fs}:=\oplus_{\ell+k=1, k\leq 2}\mathcal{D}^{k,\ell}(C)=\mathcal{D}^{0,1}(C)\oplus \mathcal{D}^{1,0}(C) \oplus \mathcal{D}^{2,-1}(C).\]
We let $\fs$ be the derived Lie algebra  $[\tilde{\fs},\tilde{\fs}]$. We will  call the Lie algebra $\fs$ the \textit{quantum dynamical symmetry}. It  can be verified that $\tilde{\fs}=\fs\oplus \C$.

\begin{theorem*}[Compare with Prop. 3.8 and Cor. 4.4 of \cite{LevasseurAndStafford2017}]
 The complex Lie algebra $\fs$
is  isomorphic to $ \mathfrak{so}(n+2,\C)$ and it generates  $\mathcal{D}(C)$ as a unital complex algebra.
   \end{theorem*}\label{Theorem1} 
For an explicit embedding of  $\mathfrak{so}(n+2,\C)$  in $ \mathcal{D}(C)$   see section 9 of \cite{LevasseurAndStafford2017} and  Proposition \ref{sss232}. Further details can be found in  \cite{Levasseur1986,LevasseurSmithStafford1988,Goncharov1982} and section 1.1 of \cite{KobayashiMano2011}.

\begin{remark*}
   We note that the Euler operator $\sum_ix_i\partial_{x_i}$ belongs to    $\tilde{\fs}$ but not to ${\fs}$. The shifted Euler operator $h:=2\sum_ix_i\partial_{x_i}+(n-2)$ does belong to $\fs$. The scalar shift $(n-2)$ arises from the natural action on half-forms. 
\end{remark*}

\subsubsection{The   isomorphism $\mathcal{D}^{0,1}(C)\longrightarrow \mathcal{D}^{2,-1}(C)$}\label{s224}
In this section we show that the form $Q$ induces  a canonical endomorphism $\Phi^Q$ of the vector space  $\mathcal{D}(C)$.  We focus on  its restriction to the subspace $\mathcal{D}^{2,-1}(C)$ of $\fs$ which will be needed later on when we define the Schr\"odinger family.

We let $B_Q$ be the bilinear symmetric form associated to $Q$. Recall that coordinate functions $\{z_i\}$ on $V$  belong to $\mathcal{O}_C$, and using the embedding   $\mathcal{O}_C\hookrightarrow \mathcal{D}(C)$, we can think of each $z_i$ as a differential operator on the cone. 
\begin{proposition*}
  There exists a linear morphism   $\Phi^Q:\mathcal{D}(C)\longrightarrow \mathcal{D}(C)$ such that
  if the form  $Q$ is given in coordinates by $\sum_{i,j=1}^nB_{ij}z_iz_j $, then 
for any $d\in \mathcal{D}(C)$,
\[\Phi^Q(d)=  \sum_{i,j=1}^nB_{ij}z_idz_j.\]
 Moreover $\Phi^Q$ is $O(Q)$-equivariant. 
\end{proposition*}
The proof is by a straightforward calculation.

 The restriction of $\Phi^Q$  to $\mathcal{D}^{2,-1}(C)$ is an isomorphism onto $\mathcal{D}^{0,1}(C)$. Before showing this we shall define  a map $\Psi^Q: \mathcal{D}^{0,1}(C)\longrightarrow \mathcal{D}^{2,-1}(C)$
 that will be proportional to the inverse of the mentioned isomorphism. For technical reasons  it  is more convenient to work with this un-normalized version $\Psi^Q$. 

We set $n=\operatorname{dim(V)}$ and  recall that $h=2\sum_ix_i\partial_{x_i}+(n-2) \in \fs$.  We let $\Box_Q$ be  the canonical  second-order differential operator on $V$ corresponding to the symmetric bilinear form $B_Q$. For example, if in coordinates $Q(v)=\sum_{i=1}^nA_iz_i^2$ then $\Box_Q=\sum_{i=1}^nA_i^{-1} (\partial_{z_i})^2$.

The quadratic form $Q$ induces an isomorphism   $\tau:V\longrightarrow V^*$. In this way we can assign to any linear functional $w\in V^*=\mathcal{D}^{0,1}(C)$ a unique vector  $v_w\in V$. We let $\partial_w^Q$   be the directional derivative corresponding to  $v_w$.

We define a linear map $\Psi^Q:V^*\longrightarrow \mathcal{D}^{2,-1}(V)$ by \[\Psi^Q(w)  =w \Box_Q-h\partial_w^Q,\quad  \forall w \in V^*.\]

\begin{lemma*}
 The linear map   $\Psi^Q$ is an isomorphism onto its image. The image  is equal to  $\mathcal{D}^{2,-1}(V)_I$, and  
\[\Phi^Q\circ \Psi^Q=(2-n)\mathbb{I}_{V^*}.\]
\end{lemma*}

 \begin{proof}
     By direct calculations we see that   $\Psi^Q(V^*)\subseteq \mathcal{D}^{2,-1}(V)_I$  and for every  $w\in V^*$, $\Phi^Q\circ \Psi^Q(w)=(2-n)w$. In addition it is easy to show (and can be deduced from \cite{Goncharov1982}) that $\operatorname{dim}(\mathcal{D}^{2,-1}(V)_I)\leq \operatorname{dim}(V)$. 
 \end{proof}
 We recall that $V^*\simeq \mathcal{D}^{0,1}(C)$ and $\mathcal{D}^{2,-1}(V)_I\simeq \mathcal{D}^{2,-1}(C)$. By abuse of notation we denote the corresponding isomorphism $\mathcal{D}^{0,1}(C)\longrightarrow \mathcal{D}^{2,-1}(C)$ also by 
 $\Psi^Q$.

In Appendix \ref{JM} we relate $\Psi^Q$ to the Jacobson-Morozov theorem.

\subsubsection{The symmetry group}\label{223}
The group $ G_1= O(Q)\times \C^{\times}$
  naturally  acts  on $V$ in a way that  preserves $C$.  As a result   $G_1$ acts on  $\mathcal{D}(V)$  and $\mathcal{D}(C)$.  
 The spaces $\tilde{\fs}$ and $\fs$ are  $G_1$ stable. 
Since $\mathcal{D}^{0,1}(V)_I$, $\mathcal{D}^{1,0}(V)_I$, $\mathcal{D}^{2,-1}(V)_I $ trivially  intersect $I\mathcal{D}(V)$, the map from Proposition \ref{s212} induces a canonical $G_1$-equivariant isomorphism $\mu$ of  complex Lie algebras  from $\tilde{\fs}$ onto its image
\[ \mu(\tilde{\fs}):=\mathcal{D}^{0,1}(V)  \oplus  \mathcal{D}^{1,0}(V)_I \oplus  \mathcal{D}^{2,-1}(V)_I.\]
We will freely use the isomorphism  $\mu$ throughout the text without explicitly mentioning it.

In addition, one can verify the following facts that we will  use later on.
\begin{proposition*} 
As representations of $G_1$ we have the following relations: 
\begin{enumerate}
\item $\mathcal{D}^{0,1}(C)\cong  \mathcal{D}^{0,1}(V)_I=V^*$.  
\item $\mathcal{D}^{1,0}(C) \cong \mathcal{D}^{1,0}(V)_I= \mathfrak{o}(Q)\oplus \C 1 \oplus \C h $.
\item $ \mathcal{D}^{2,-1}(C)\cong \mathcal{D}^{2,-1}(V)_I \cong V $.  
\item $[\mathcal{D}^{0,1}(C),\mathcal{D}^{2,-1}(C)]\cong \mathfrak{o}(Q)\oplus \C h$ 
\end{enumerate}
In particular the decomposition of $\fs$ into irreducibles of  $G_1$ is  given by 
\[\fs \cong   V\oplus   \mathfrak{o}(Q)  \oplus \C      \oplus V^*.\] 
\end{proposition*}
We remark that a description of $ \mathcal{D}(C)$ is essentially  given in \cite{Goncharov1982}. We will also give an explicit description in the following section.

 \subsubsection{A spanning set for $\tilde{\fs}$    }\label{s2.4}
 In this section we specify a spanning set for $\tilde{\fs}$ that can be used to prove several results that are stated in the following sections.

 Since $\Psi^Q: \mathcal{D}^{0,1}(C)\longrightarrow \mathcal{D}^{2,-1}(C)$, which  was introduced in section \ref{s224}, is an isomorphism, clearly \begin{eqnarray}\nonumber
&&\mathcal{D}^{0,1}(C)\cong \mathcal{D}^{0,1}(V)=\{\tau(v)|v\in V \}, \\ \nonumber
&&\mathcal{D}^{2,-1}(C)\cong  \mathcal{D}^{2,-1}(V)_I=\{\Psi^Q(\tau(v))|v\in V\}.
\end{eqnarray}
For $u,v\in V$ we define, 
\[L_{u,v}:=\tau(u)\circ \partial_v-\tau(v)\circ \partial_u\in \mathcal{D}^1(V),  \]
where $\partial_v$ is the directional derivative corresponding to $v$ and $\tau:V\longrightarrow V^*$ is the isomorphism that the   quadratic form $Q$ induces.  Note that $L_{u,v}\neq 0$ if and only if $\{u,v\}$ is linearly independent. By direct calculations we can show that 
\begin{eqnarray}\nonumber
&&\mathcal{D}^{1,0}(C)\cong \mathcal{D}^{1,0}(V)_I=\operatorname{span}_{\C}\{L_{u,v}|u,v\in V \}\oplus \C h \oplus \C.
\end{eqnarray}
We remark that $\operatorname{span}_{\C}\{L_{u,v}|u,v\in V \}$ is canonically identified with the Lie algebra $\mathfrak{o}(Q)$.

Since $\tilde{\fs} = \mathcal{D}^{0,1}(C)\oplus \mathcal{D}^{1,0}(C)\oplus \mathcal{D}^{2,-1}(C)$ we have specified a generating set for $\tilde{\fs}$.  The commutation relations of vectors in the specified spanning set are easy to calculate. They are given in Appendix \ref{crs}.

\subsubsection{Compatibility with real forms}\label{CmpRF}
The results presented in this section (Section \ref{se2}) for complex quadratic spaces hold with appropriate changes to their real forms.  Here we would like to point out  two facts.

If $(\mathcal{V},q)$ is a real form of $(V,Q)$, the algebra of real algebraic differential operators on the corresponding real cone (as a real algebraic variety) is a real form of the algebra $\mathcal{D}(C)$ of complex  algebraic differential operators on the complex cone. This induces a real form $\tilde{\fs}(\R)$   of the complex Lie algebra $\tilde{\fs}\subset \mathcal{D}(C)$. 
If the real form $(\mathcal{V},q)$  has 
signature $(\ell,m)$  the induced real form on $\fs$ is isomorphic to $\mathfrak{o}(\ell+1,m+1)$. This is    explained  in  Appendix \ref{s2.5}.

In section \ref{canreal} we introduce another  real form $\tilde{\fs}^*(\R)$ of $\tilde{\fs}$ which will be the one we shall use most. The two real forms are different but isomorphic. 

 We  remark that a real form of a complex vector  space $V$ (or a quadratic space or a Lie algebra) is the same thing as an antiholomorphic involution $\sigma$ of $V$.

\section{Quadratic spaces and reductive dual pairs}\label{sect3}
In this section we  explain how a pointed real quadratic space gives rise to a reductive dual pair in $\fs$. This will play an important role for the  Schr\"odinger family that is    introduced  in Section \ref{502}. 
\subsubsection{Pointed  Lorentzian quadratic spaces }
Our convention for a Lorentzian quadratic form  on an  $n$-dimensional vector space $\mathcal{V}$ is such that the signature of $q$ is $(n-1,1)$.

\begin{defn*}
\textit{A pointed Lorentzian quadratic space} (PLQS) is  a triple $(\mathcal{V},q,w)$ such that $(\mathcal{V},q)$ is a finite-dimensional (real) Lorentzian quadratic space and $w$ is a linear functional on $\mathcal{V}$ of norm $-1$. That is, under the dual form $q^*$ on $\mathcal{V}^*$, $q^*(w)=-1$.
\end{defn*}

\begin{remark*}
 The form $q$ induces a canonical isomorphism  $\mathcal{V}\cong \mathcal{V}^*$ and in this way the dual space also carries a structure of a PLQS.     
\end{remark*} 

Throughout the text we will also work with the complexification $(V,Q)$ of $(\mathcal{V},q)$.

\subsubsection{The reductive dual pair symmetry of a PLQS}\label{PLQS}
In this section, with a PLQS we associate a reductive dual pair in the quantum dynamical symmetry $\fs$.

Throughout this paragraph, we let $(\mathcal{V},q,w)$ be   an $n$-dimensional  PLQS with $n>2$. 

Howe introduced the notion of a reductive dual pair, see e.g., \cite{MR789290}. Given a reductive Lie algebra $\fg$ and a pair of reductive Lie subalgebras $(\mathfrak{k},\mathfrak{l})$ of $\fg$, the pair is called a \textit{reductive dual pair} if $\mathfrak{k}$ is equal to the centralizer of $\mathfrak{l}$ in $\fg$ and $\mathfrak{l}$ is equal to the centralizer of $\mathfrak{k}$ in $\fg$.

The orthogonal group $O(Q)$ acts on  the dual space $V^*$. 
The stabilizer ${K}_{w}:=\operatorname{stab}_{O(Q)}(w)$ is isomorphic to $O(n-1,\C)$. 
We let $\mathfrak{k}_{w}\subset \mathcal{D}^{1,0}(C)$ be the complex Lie algebra of $K_{w}$.  We
define $\mathfrak{l}_{w}:=\operatorname{Cent}_{\fs}(\mathfrak{k}_{w})$.

We also note that the real Lie algebra $\mathfrak{o}(q)$ is given by  $\operatorname{span}_{\R}\{L_{u,v}|u,v\in \mathcal{V}\}$ and isomorphic to $\mathfrak{so}(n-1,1)$.
 By Proposition \ref{s2.5}  the real form $\fs(\R)$ is isomorphic to $\mathfrak{so}(n,2)$. We 
define $\mathfrak{k}_{w}(\R):=\mathfrak{k}_{w}\cap \fs(\R)$ and $\mathfrak{l}_{w}(\R):=\mathfrak{l}_{w}\cap \fs(\R)$.

\begin{proposition*} 
With the above-mentioned notation.
\begin{enumerate}
    \item The pair $(\mathfrak{k}_{w},\mathfrak{l}_{w})$ is a reductive dual pair in $\fs$. It is isomorphic to $(\mathfrak{o}(n-1,\C),  \mathfrak{sl}_2(\C))$ in $\mathfrak{so}(n+2,\C)$.  
    \item The pair $(\mathfrak{k}_{w}(\R),\mathfrak{l}_{w}(\R))$ is a real reductive dual pair in $\fs(\R)$ that is isomorphic to $(\mathfrak{so}(n-1),\mathfrak{sl}_2(\R))$  in $\mathfrak{so}(n,2)$. 
\end{enumerate}
\end{proposition*}
The proof is by direct calculations and is given in Appendix \ref{l}.

\section{The canonical Hilbert space $\boldsymbol{H}$.}\label{4}
In this section, starting with an $n$-dimensional  PLQS $(\mathcal{V},q,w)$ we construct a canonical Hilbert space $\boldsymbol{H}$.

Throughout  this section we will assume that we are given an $n$-dimensional  PLQS  $(\mathcal{V},q,w)$ with $n$ being an even integer such that $n>2$.

\subsection{The  canonical Hilbert space}\label{sub4.2}
In this section, starting with a PLQS $(\mathcal{V},q,w)$ we introduce a canonical positive density $|\omega|$ on the corresponding  real cone $C_0(\R)$. We define the canonical Hilbert space $\boldsymbol{H}$ of our PLQS to be the corresponding  $L^2$ space.  
 \subsubsection{An orientation of a complex quadratic space}

\begin{defn*}
 Let  $(V,Q)$ be a complex $n$-dimensional   quadratic space.
An \textit{orientation} of    $(V,Q)$ is  a top form  $\Omega$ on $V$  such that for any   set  of $n$ vectors $\{v_1,v_2,...,v_n\}\subset V$, 
\[\left(\Omega(v_1,v_2,...,v_n)\right)^2=\det \left(G\right),\]
where $G$ is the Gram matrix given by  $G_{i,j}=B(v_i,v_j)$, and $B$ is the bilinear symmetric form corresponding to $Q$. 

  An \textit{oriented quadratic space} is a quadratic space with a choice of orientation. 
\end{defn*}
It is easy to show that an orientation exists and is unique up to a sign.

\begin{example*}
If $Q$ is the quadratic form on $\C^n$ that is given in coordinates by $Q=z_1^2+z_2^2+...+z_n^2$,
 the two orientations of $(\C^n,Q)$ are given by 
    \[\pm dz_1\wedge dz_2 \wedge ... \wedge dz_n.\]
\end{example*}
 \subsubsection{An orientation of a real quadratic space}

\begin{defn*}
 An \textit{orientation} 
of  a real quadratic space $(\mathcal{V},q)$ is  an orientation  $\Omega$ of its complexification $(V,Q)$.
 
 An \textit{oriented real quadratic space} is a quadratic space with a choice of orientation (for its complexification). 
\end{defn*}
 
\begin{example*}
If $q$ is the quadratic form on $\R^n$ with signature $(\ell,m)$ that is given in coordinates by $q=x_1^2+...+x_{\ell}^2-x_{\ell+1}^2-...-x_{\ell+m}^2$. 
 The two orientations of $(\R^n,q)$ are given by 
    \[\pm (i)^{m} dx_1\wedge  ... \wedge dx_n.\]
\end{example*}

\subsubsection{The canonical top form on  $C_0$ }
Let $(V,Q,\Omega)$ be an oriented complex quadratic space.

We denote by $C_0$ the variety of smooth points of the corresponding complex cone $C$. Note that $C_0:=\{v\in V| Q(v)=0 \}\setminus\{0\}$. 

\begin{claim*}
    There is a unique algebraic differential $(n-1)$-form $\omega$ on $C_0$ characterized by the following property. 
    For any  extension $\widetilde{\omega}$ of $\omega$ to a rational  $(n-1)$-form on $V$,  the rational function $f$ on $V$ that is defined by $dQ\wedge \widetilde{\omega}=f \Omega$ is  equal to $1$ on $C_0$.
\end{claim*} 
The form $\omega$ can be  constructed locally as follows.    For every $p\in C_0$, there is an open neighborhood $U_p$ of $p$ in $V$ and an $(n-1)$-form $\eta_p$ on $U_p$ such that $dQ\wedge \eta_p= \Omega$ on $U_p$. On $U_p\cap C_0$, 
$\omega|_{U_p\cap C_0}=\eta_p|_{U_p\cap C_0}$.

Hence an  orientation $\Omega$  on  $(V,Q)$ gives rise to a canonical top form  $\omega$ on $C_0$.   Clearly $\omega$ is $SO(Q)$-invariant.

\subsubsection{The canonical density on a  Lorentzian space}\label{423}
In this section we show that a Lorentzian quadratic space carries a canonical density. 
 
We let $(\mathcal{V},q)$ be a real Lorentzian oriented  quadratic space with complexification $(V,Q)$ and orientation $\Omega$.
The real cone $C_0(\R):=C_0\cap \mathcal{V}$ is a smooth manifold. Restriction of the canonical form  $\omega$ to $C_0(\R)$ gives a complex-valued smooth $(n-1)$-form. 
  
Taking the absolute value defines a canonical positive density $|\omega|$ on $C_0(\R)$.  It is independent of the choice of orientation.

Assuming that in terms of local coordinates  $\{z_1,z_2,...,z_n\}$ on $V$ the form $Q$ is given by ${Q}=z_1^2+...+z_{n-1}^2-z_n^2$, 
then the form $\omega$  is given on the open set $U=\{p\in C_0|z_{n}(p)\neq 0 \}$ by 
\[\omega = \frac{\pm i}{2z_n} dz_1\wedge ...\wedge dz_{n-1}.\]

The real cone $C_0(\R)$   has a decomposition into connected components 
\[C_0(\R)=C_+(\R)\sqcup C_-(\R), \]
which is determined by the sign of $z_n$.

On   each of the half cones $C_+(\R),C_-(\R)$ the canonical positive density $|\omega|$   is given by 
\[ |\omega|=  \frac{1}{2|z_{n}|} |dz_1\wedge ...\wedge dz_{n-1}|.\]
 We can express $z_n$ on $C_{\pm}(\R)$ as $\pm\sqrt{z_1^2+...+z_{n-1}^2}$.

\subsubsection{The  canonical Hilbert space $\boldsymbol{H}$}\label{canhil}

We define the canonical  Hilbert space $\boldsymbol{H}$ associated with the PLQS $(\mathcal{V},q,w)$ to be  the  $L^2$ space corresponding to ${C}_{0}(\R)$ with the density  $|\omega|$, explicitly $\boldsymbol{H}=L^2(C_0(\R),|\omega|)$. The Hilbert space does not depend on the orientation of $(V,Q)$.  

The decomposition $C_0(\R)=C_+(\R)\sqcup C_-(\R)$ induces a decomposition of Hilbert spaces 
\begin{eqnarray}
  &&\boldsymbol{H}= \boldsymbol{H}_{+} \oplus \boldsymbol{H}_-.
 \end{eqnarray}

\section{The  distinguished self-adjoint action of the algebra $\mathcal{D}(C)$ on $\boldsymbol{H}$}\label{ddd}
In this section we describe  a distinguished "Schwartz" subspace $\mathcal{S}(\boldsymbol{H})$ of $\boldsymbol{H}$ that carries    a self-adjoint action of the algebra $\mathcal{D}(C)$. There are different methods to construct this Schwartz subspace. We first construct the  group $G=G(\mathcal{V})$ of quantum symmetries of our Lorentzian quadratic space $(\mathcal{V},q)$.  This group is an extension of $O(\mathcal{V})$.   We describe a canonical extension of the representation of $O(\mathcal{V})$ on    $\boldsymbol{H}$ to a unitary representation $\pi$ of $G$ on $\boldsymbol{H}$.
We construct 
$\mathcal{S}(\boldsymbol{H})$
as the subspace  of smooth vectors for $\pi$,   which we think of as a suitable  Schwartz subspace of $\boldsymbol{H}$.

In Subsection \ref{sub4.1}  we start by considering a general notion of an action of a $*$-algebra on a Hilbert space and in particular symmetric and   self-adjoint actions. 

In Subsection \ref{sub4.3} we  introduce the intermediate subspace  $C_c^{\infty}(C_0(\R))$ of smooth compactly supported functions on the cone. This subspace is the domain of  a symmetric action of the algebra $\mathcal{D}(C)$.

In Subsection \ref{seext} we recall that $\boldsymbol{H}$ carries a distinguished unitary irreducible representation  $\pi$ of the group $G$. We show that the canonical subspace $\mathcal{S}(\boldsymbol{H})$  consisting of smooth vectors for $\pi$  carries a self-adjoint action of  $\mathcal{D}(C)$ which is an extension of the above-mentioned symmetric action.

Throughout  this section we will assume that we are given an $n$-dimensional  PLQS  $(\mathcal{V},q,w)$ with $n$ being an even integer such that $n>2$.
\subsection{Action of a $*$-algebra on a Hilbert space}\label{sub4.1} 
In this section we recall the notion of a partial  action of a $*$-algebra on a Hilbert space and discuss some of its  properties.  We essentially follow \cite{MR283580}  (also see \cite{MR333743,MR3722559}).

\subsubsection{A partial action of a $*$-algebra}\label{dense}
Recall that   a complex $*$-algebra $A$ is an associative unital algebra over $\C$ that is equipped with  an involution $*:A\longrightarrow A$ which is anti-linear  anti-homomorphism. Explicitly
$\forall a,b\in A$ and $\alpha\in \C$,
\begin{eqnarray}\nonumber
&&(a^*)^*=a, (a+b)^*=a^*+b^*, (ab)^*=b^*a^*, 1^*=1, (\alpha a)^*=\overline{\alpha}a^*.
\end{eqnarray}
The following definition is a generalization of the theory   of unbounded operators on a Hilbert space.

Let $A$ be a complex $*$-algebra  and $\boldsymbol{H}$  a complex Hilbert space.

\begin{defn*}
A \textit{partial action} of  $A$  on   $\boldsymbol{H}$ is a pair $(\rho,\mathcal{W})$ where 
$\mathcal{W}$ is a linear subspace  of $\boldsymbol{H}$, and 
$\rho: A\longrightarrow \operatorname{End}(\mathcal{W})$ is an action of the unital associative algebra $A$ on   $\mathcal{W}$.

We  call $\mathcal{W}$ the \textit{domain} of the action.

For two partial actions $(\rho,\mathcal{W})$ and $(\rho',\mathcal{W}')$ of the same $*$-algebra on the same Hilbert space, we  write $\rho\preceq \rho'$ whenever $\rho'$ is an extension of $\rho$.
\end{defn*}

\subsubsection{A dense action of a $*$-algebra and its adjoint}
A partial action $(\rho,\mathcal{W})$ is called a \textit{dense action} if the domain $\mathcal{W}$ is dense in $\boldsymbol{H}$. 
For such a dense action, we define another partial action  of $A$ on $\boldsymbol{H}$,  the \textit{adjoint  action} $(\rho^{*}, \mathcal{W}^{*})$ as follows. 
The  domain $\mathcal{W}^{*}$
consists of all $v\in \boldsymbol{H}$ for which for every $a\in A$, the linear functional on $\mathcal{W}$ that is given by
\[x\longmapsto \langle \rho(a)x,v\rangle\]
is bounded by $C(a,v)\| x\|$ for some positive constant  $C(a,v)$. The adjoint action  $\rho^{*}$ is uniquely defined  by 
\[\langle \rho^{*}(a)v, x\rangle=\langle v, \rho(a^*)x\rangle\]
for every $a\in A$, $v\in \mathcal{W}^{*}$ and $x\in \mathcal{W}$. 
Obviously  $\rho\preceq \rho'$ implies that $(\rho')^{*}\preceq \rho^{*}$.

\subsubsection{Symmetric, self-adjoint and essentially self-adjoint actions}\label{413}

\begin{defn*} \hspace{3cm}
\begin{enumerate}
    \item  A dense action $\rho$ of $A$  on a complex Hilbert space $\boldsymbol{H}$ is \textit{symmetric} if its adjoint $\rho^{*}$ is an extension of  $\rho$, that is, if $\rho \preceq \rho^{*}$.  In other words  this means that 
\[\langle \rho(a)x,v\rangle=\langle x,\rho (a^*)v\rangle,\]
for   $a\in A$, $v,x\in \mathcal{W}$.
\item A symmetric dense action $\rho$  is \textit{self-adjoint} if it is equal to its adjoint, that is, if  
$\rho = \rho^{*}$.
\item A symmetric dense action $\rho$  is \textit{essentially self-adjoint} if $\rho^{*}$ is self-adjoint.
\end{enumerate}
 
\end{defn*}

\begin{proposition*}
     Let   $\rho$ be a symmetric dense action of $A$  on a complex Hilbert space $\boldsymbol{H}$.   The action $\rho$ is essentially self-adjoint if and only if its adjoint is symmetric.
     Moreover, in that case for every symmetric extension $\rho'$ of $\rho$ we have $(\rho')^{*}=\rho^{*}$.
\end{proposition*}

The proof is straightforward.

\subsubsection{An important  example of a self-adjoint action}\label{exa}

A unitary representation of a Lie group $G$ gives rise to a dense self-adjoint action of the universal enveloping algebra  $A=\mathcal{U}(\fg)$ where $\fg$ is the complexification of the Lie algebra $\operatorname{Lie}(G)$.

Let 
$(\pi,G,\boldsymbol{H})$  
be a unitary representation of a Lie group $G$ on a Hilbert space $\boldsymbol{H}$.  Let $\boldsymbol{H}^{\infty}$ be the subspace of  smooth vectors in $\boldsymbol{H}$ (recall that a vector is smooth if the  corresponding orbit map is smooth). The universal enveloping algebra $\mathcal{U}(\fg)$  is a complex $*$-algebra with respect to the unique involution $*:\mathcal{U}(\fg)\longrightarrow \mathcal{U}(\fg)$ that is minus the identity on  $\operatorname{Lie}(G)$.
\begin{claim*}
    The natural action $\pi^{\infty}$ of $\mathcal{U}(\fg)$ on $\boldsymbol{H}^{\infty}$ is a  dense self-adjoint action.
\end{claim*}
We will prove this claim in Appendix \ref{sv}. 
 Further  information on the underlying smooth space of a representation of a Lie group can be found in \cite{BK2014}.

\subsection{The canonical symmetric  action of $\mathcal{D}(C)$ on $\boldsymbol{H}$}\label{sub4.3}
 As a first step for  constructing   the Schwartz subspace $\mathcal{S}(\boldsymbol{H})$, in  this subsection  we    introduce a canonical symmetric action of the algebra $\mathcal{D}(C)$ on $\boldsymbol{H}$ with the domain of compactly supported smooth functions on the cone.

\subsubsection{The symmetric action}

The   algebra $\mathcal{D}(C)$  naturally acts on $C^{\infty}(C_0(\R))$, the space of smooth  complex-valued  functions on $C_0(\R)$. We denote this action by $\rho^{\infty}$.
 
The subspace  $C_c^{\infty}(C_0(\R))$  of     compactly supported functions   is a  dense $\rho^{\infty}$-stable   subspace  of  $C^{\infty}(C_0(\R))$.

We  denote the action of the algebra  $\mathcal{D}(C)$ on $C_c^{\infty}(C_0(\R))$  by $\rho_c$. 
\begin{claim*}
For every $d\in \mathcal{D}(C_0)$ there is a unique $d^*\in \mathcal{D}(C_0)$ such that 
  \begin{equation}\nonumber
 \langle \rho_c(d)f_1,f_2 \rangle=\langle f_1,\rho_c(d^*)  f_2 \rangle , \quad \forall f_1,f_2\in C^{\infty}_c(C_0(\R)). 
\end{equation}   
\end{claim*}
\begin{proof}
    Differential operators on $C_0$ are described by their restrictions to affine
open subsets. We will construct a family $\{U_P\}$ of such subsets with the following
properties
\begin{enumerate}
    \item Every $U_P$ is a dense affine open subset defined over $\R$.
    \item The family $\{U_P\}$  is closed under intersections and covers $C_0$.
    \item For every set $U_P$,  the algebra of differential operators $\mathcal{D}(U_P)$ satisfies the following condition; 
    
    $(*)$ For every $d\in \mathcal{D}(U_P)$ there is a unique $d^*\in \mathcal{D}(U_P)$ such that 
  \begin{equation}\nonumber
 \langle \rho_c(d)f_1,f_2 \rangle=\langle f_1,\rho_c(d^*)  f_2 \rangle , \quad \forall f_1,f_2\in C^{\infty}_c(U_P(\R)). 
\end{equation}  
\end{enumerate}

Since $\mathcal{D}(C_0)$   can be glued  out of the  algebras  $\{\mathcal{D}(U_P)\}$,  a standard gluing argument implies the claim.

Let $P$ be a non-zero polynomial on $C$ defined over $\R$ such that $P(0) = 0$. We
define the set  $U_P$ by $U_P=C\setminus \{\text{zeroes of }P\}$.
Conditions (1.) and (2.) are obviously satisfied. Uniqueness in condition (3.) follows from the fact that 
$U_P(\R)$ is Zariski dense
in $U_P$. Existence of an operator  $d^*$ follows from the fact that the algebra $\mathcal{D}(U_P)$
is generated by vector fields given by elements in $\operatorname{Lie}(O(q))$  and by regular  functions
 $\mathcal{O}(U_P)$. For all such  generators $d$ the existence of the operator $d^*$ is clear.
\end{proof}

\begin{remark*}
Note that since $\mathcal{D}(C)=\mathcal{D}(C_0)$ the claim says that there is a unique $*$-algebra structure on
$\mathcal{D}(C)$ such that the action of $\mathcal{D}(C)$ on $C^{\infty}_c(C_0(\R))$ is a symmetric dense action
$\rho_c$ of $\mathcal{D}(C)$ on $\boldsymbol{H}$.
\end{remark*}

\subsubsection{The induced real structure on the Lie algebra  $\fs$ }\label{canreal}
 
In this subsection we explain how the symmetric action of  $\mathcal{D}(C)$ induces a real structure on $\fs$ which has a natural reductive dual pair.

 The Lie algebra $\tilde{\fs}$ is stable under the involution $*$. On this complex Lie algebra $\tilde{\fs}$ and its subalgebra $\fs$,  the  negative of the involution $*$ defines an antiholomorphic involution 
 \[d\longmapsto -d^*.\]
The fixed-point subalgebra $\fs^{*}(\R):=\{d\in \fs|d^*=-d\}$ is a real form of $\fs$ different from the real
form $\fs(\R)$ introduced in Subsection \ref{CmpRF}.
We will show that these two real forms are isomorphic, see details in appendix 
\ref{s2.5} and appendix \ref{a823}.

Among these two real forms, the main one that we use in this paper is 
$\fs^*(\R)$.

 \subsection{The distinguished self-adjoint extension }\label{seext}
In this Subsection our goal is to describe a canonical self-adjoint extension  $(\rho,\mathcal{S}(\boldsymbol{H}))$
 of the symmetric action $(\rho_c,C_c^{\infty}(C_0(\R)))$. Namely, we describe a subspace  $\mathcal{S}(\boldsymbol{H})\subset \boldsymbol{H}$  in-between $C_c^{\infty}(C_0(\R))$ and $C^{\infty}(C_0(\R))$, invariant with respect to the action of the algebra $\mathcal{D}(C)$, such that the action $(\rho,\mathcal{S}(\boldsymbol{H}))$ is self-adjoint. 

   There are several descriptions of this space. We will present a construction that
   realizes $\mathcal{S}(\boldsymbol{H})$  as the space of smooth vectors  in a representation of a “large group” $G=G(\mathcal{V})\supset O(\mathcal{V})$  on the  Hilbert space $\boldsymbol{H}$.

\subsubsection{A canonical extension of the representation of $O(\mathcal{V})$ on  $\boldsymbol{H}$}\label{dis}
In this section we construct the group $G=G(\mathcal{V})$ of quantum symmetries of $(\mathcal{V},q)$ which is a group  extension  of $O(\mathcal{V})$. We describe a distinguished representation  $\pi$ of $G$  on the canonical Hilbert space  $\boldsymbol{H}$.

\textbf{The extension $G$ of $O(\mathcal{V})$:} Let  $(\R^{1,1}=\R^2,q_{1,1})$ be the standard hyperbolic plane with the quadratic form $q_{1,1}(x,y)=2xy$. This is a    two-dimensional Lorentzian quadratic space. The standard basis $v_1:=(1,0)$, $v_2:=(0,1)$  consists  of isotropic vectors satisfying $B_{1,1}(v_1,v_2)=1$, where $B_{1,1}$ is the corresponding symmetric bilinear form.

We set $\widetilde{\mathcal{V}}:=\mathcal{V}\oplus \R^{1,1}$ and $\widetilde{q}=q\oplus q_{1,1}$ and we define the quantum symmetry group $G=G((\mathcal{V},q))$ of $(\mathcal{V},q)$ to be  the orthogonal group $G((\mathcal{V},q))=O(\widetilde{\mathcal{V}},\widetilde{q})$. We denote by $L$ the canonical copy of $O(\mathcal{V})$ in $G$ consisting of all elements that act as the identity map on $\R^{1,1}\subset \widetilde{\mathcal{V}} $,  and call it \textit{the Lorentz subgroup of $G$}.

We define \textit{the Poincar\'e subgroup $P=P_{v_1}$ of $G$} to be the stabilizer of $v_1$ in $G$.   
Clearly 
\[L\subset P \subset G .\]

Our next goal is to show that there exists a canonical extension of the representation of $L$ on $\boldsymbol{H}$ to a representation of $G$.

We will  do it in two steps. First we extend it to a representation of $P$, and then to a representation of  $G$. 

 \textbf{The extension of the representation to $P$:}
We let $N$ be the unipotent radical of 
${P}$. It follows that ${P}=LN$ as a semidirect product with $N$ a normal subgroup isomorphic to $\mathcal{V}^*$.  Explicitly, an $L$-equivariant isomorphism of abelian groups  $f_{v_2}:N\longrightarrow \mathcal{V}^*$ is given  by 
 \[\left(f_{v_2}(n)\right)(v):=B(nv_2,v) ,\]
for any $n\in N$ and $v\in \mathcal{V}$, and 
where $nv_2$ stands for the application of the linear transformation $n$ to $v_2$.

 The  representation of $L$ on  $\boldsymbol{H}$    extends to a representation $\pi=\pi_n$ of $P$, by defining $\pi(n)$ for $n\in N$ to be the  operator of multiplication by the function on $C_0(\R)$ given by  $v\longmapsto e^{iB(nv_2,v)}$.

 By Mackey theory this unitary representation of $P$ is irreducible,  see also \cite[Prop. 3.3]{KOBAYASHI2003551}.
 
 \textbf{The extension of the representation to $G$:} The extension of $\pi_{n}=\pi:{P}\longrightarrow \mathcal{U}\left(\boldsymbol{H})\right)$ to a representation of $G$ is based on the fact that $G$ has a distinguished representation $\widetilde{\pi}$ of minimal Gelfand-Kirillov dimension, namely the local theta lift of the trivial representation of $SL_2(\R)$, see  \cite{MR1457244}.  Further info on the  theta correspondence can be found in \cite{MR546602,MR985172,MR986027}.

\begin{theorem*}
For any $n\in 2\N$ with $n\geq 4$, the representation 
$\pi=\pi_{n}:{P}\longrightarrow \mathcal{U}\left(\boldsymbol{H}\right)$
can be extended uniquely  to a unitary irreducible representation of $G$ that is isomorphic to $\widetilde{\pi}_{n}$.
\end{theorem*}

\begin{proof}

Uniqueness  follows from irreducibility of $\pi:{P}\longrightarrow \mathcal{U}\left(\boldsymbol{H}\right)$ together with Schur's lemma.
   
       The  hard part of the proof  is to show the  existence of such an  extension. This is proven in \cite[Thm. 3]{KOBAYASHI2003551}. 
\end{proof}

 
\begin{remark*}
1. The existence of the extension  of  $\pi$ from ${P}$ to $G$ is highly non-trivial. One piece of evidence for
this is that in the action of the element $-I\in G$ there is a non-obvious sign-rule. Explicitly, 
 the action  is given by 
 $$\pi(-I)=(-1)^{(n-2)/2}.$$

2. For $n>4$ the representation $\widetilde{\pi}_n$ is the unique minimal representation of $G$, see \cite{KobayashiMano2011}. 

3. In this paper we shall only use  the Schr\"{o}dinger model of $\widetilde{\pi}$ which is realized on the canonical Hilbert space $\boldsymbol{H}$. 
Other important realizations of $\widetilde{\pi}$ exist. Several of them are constructed  in \cite{KobayashiMano2011}. Also see  \cite{MR2020550,MR1108044,MR1172839}.
\end{remark*}

\subsubsection{The smooth  subspace of $\pi$ as a self-adjoint extension }\label{542}

We let $\pi:G\longrightarrow \mathcal{U}(\boldsymbol{H})$ be the distinguished unitary irreducible representation on the canonical Hilbert space that was discussed above in Subsection \ref{dis}, and let $\mathcal{S}(\boldsymbol{H})$ 
be the canonical subspace of smooth vectors of $\pi$ (the Schwartz space of $\boldsymbol{H}$). In Example \ref{exa}
 this space was denoted by $\boldsymbol{H}^{\infty}$.
 \begin{proposition*}
    \begin{enumerate}
        \item $  C^{\infty}_c(C_0(\R)) \subseteq \mathcal{S}(\boldsymbol{H}) \subseteq C^{\infty}(C_0(\R)) $.
        \item $\mathcal{S}(\boldsymbol{H}) $ is invariant under the action $\rho^{\infty}$ of $\mathcal{D}(C)$.
        \item The action $\rho=\rho^{\infty}$  of the algebra $\mathcal{D}(C)$ on $\mathcal{S}(\boldsymbol{H})$ is self-adjoint. 
    \end{enumerate}
\end{proposition*}
Note that in particular the proposition says that the action of the algebra $\mathcal{D}(C)$ on $\mathcal{S}(\boldsymbol{H})$ is a self-adjoint extension of the action on $C^{\infty}_c(C_0(\R))$. 
\begin{proof}
    The first part of the proposition is proved in    \cite[Sec. 2.5]{KobayashiMano2011}. From the explicit formulas by which $\operatorname{Lie}(G)$ acts on $\mathcal{S}(\boldsymbol{H})$ via differential operators on $C_0(\R)$ see,
e.g., \cite[Sec. 2.3]{KobayashiMano2011}, it is evident 
 that the self-adjoint  action $\pi^{\infty}:\mathcal{U}(\fg)\longrightarrow \operatorname{End}(\mathcal{S}(\boldsymbol{H}))$ where $\fg$ is the complexification of the Lie algebra $\operatorname{Lie}(G)$ (as in  Subsection \ref{exa} ) 
 factors through a quotient surjective map $j:\mathcal{U}(\fg)  \twoheadrightarrow  \mathcal{D}(C)$.  
 As a result we obtain an action of $\mathcal{D}(C)$ on $\mathcal{S}(\boldsymbol{H})$ and in particular $\mathcal{S}(\boldsymbol{H})$ is invariant under the action $\rho^{\infty}$. We denote the obtained action on $\mathcal{S}(\boldsymbol{H})$  by $\rho$.

Since $j$  is a morphism of $*$-algebras and $\pi^{\infty}$ is self-adjoint, $\rho$ is self-adjoint.
\end{proof}
\begin{remark*}
 1. It is easy to directly calculate $j$ on $\operatorname{Lie}({P})$. In addition, the map $j$ is $(O(Q)\times\C^{\times})$-equivariant. Using Proposition \ref{sss232}  it follows that there is a unique $(O(Q)\times\C^{\times})$-equivariant extension of $j|_{\operatorname{Lie}({P})}$ to the complexification of $\operatorname{Lie}(G)$. In addition $j(\operatorname{Lie}(G))$ is equal to the real form $\fs^*(\R)$ introduced in Subsection  \ref{canreal}.

 2. The decomposition  $\boldsymbol{H}= \boldsymbol{H}_{+} \oplus \boldsymbol{H}_-$ is stable under the  connected component of the identity $G_0\subset G$,  see \cite[p.19]{KobayashiMano2011}. This implies that $\mathcal{S}(\boldsymbol{H})= \mathcal{S}(\boldsymbol{H})_{+} \oplus \mathcal{S}(\boldsymbol{H})_-$, where $\mathcal{S}(\boldsymbol{H})_{\pm}:=\boldsymbol{H}_{\pm}\cap \mathcal{S}(\boldsymbol{H})$ and moreover   $\mathcal{S}(\boldsymbol{H})_{\pm}$ are stable under $\rho$. In addition, letting $\mathcal{S}(\boldsymbol{H}_{\pm})$ be the subspace of smooth vectors in  representation $\boldsymbol{H}_{\pm}$ of $G_0$, we have  $\mathcal{S}(\boldsymbol{H}_{\pm})=\mathcal{S}(\boldsymbol{H})_{\pm}$.
\end{remark*}

\section{The model}\label{se5}
In this section we suggest a new model for the quantum-mechanical system of the hydrogen atom. 
\subsection{The main ingredients}
In physics, a quantum-mechanical system such as the hydrogen atom is  described by a Hilbert space and a distinguished densely defined (unbounded) self-adjoint operator called the Schr\"odinger operator. The spectrum of the Schr\"odinger operator and the expansion of a given state in terms of  eigenstates of the Schr\"odinger operator determine the time-evolution of the system via the Schr\"odinger equation. 
Other physically relevant (measurable) quantities are given via densely defined (unbounded) symmetric operators (these are called observables).

In the following subsections we list the four ingredients of our suggested model.  Our model is associated with 
an $n$-dimensional PLQS $(\mathcal{V},q,w)$,  and a real positive number $\kappa$.  We 
fix  such $(\mathcal{V},q,w)$  and $\kappa$ throughout this section.  
\subsubsection{The underlying  Hilbert space}
The first ingredient is the canonical Hilbert space associated with the $n$-dimensional  PLQS  $(\mathcal{V},q,w)$
\[\boldsymbol{H}=L^2({C}_{0}(\R),|\omega|),\] 
 as in Section \ref{canhil}.

 \subsubsection{The algebra of operators}
 The second ingredient is  
 the $*$-algebra of differential operators on the cone,
\[\mathcal{D}(C).\]

\subsubsection{The self-adjoint dense  action}
The third ingredient is
 the self-adjoint dense  action $\rho$ of the $*$-algebra $\mathcal{D}(C)$ on $\boldsymbol{H}$ with domain   $\mathcal{S}(\boldsymbol{H})$.

 Note that the algebra $\mathcal{D}(C)$ and its  canonical self-adjoint action $(\rho,\mathcal{S}(\boldsymbol{H}))$ encode the $G\simeq O(n,2)$ symmetry of the system. 

    The   subspace $\mathcal{S}(\boldsymbol{H})$ replaces the boundary conditions  that are typically used  in physics for the quantum-mechanical system of the hydrogen atom.

\subsubsection{The Schr\"odinger family}\label{502} 
The fourth ingredient is a one-parameter family  in $\mathcal{D}(C)$ called the Schr\"o\-dinger family.

Recall that with the   linear functional $w$ given in our PLQS $(\mathcal{V},q,w)$ we associate a second-order differential operator on the cone $\Psi^Q(w)$. 
 
The operators 
\[e=e_w:=iw,\quad h, \quad  f=f_w:=i\Psi^Q(w),  \] 
form  an $\mathfrak{sl}_2$-triple in $\fs\subset \mathcal{D}(C)$ satisfying the standard $\mathfrak{sl}_2$ commutation relations:
 \begin{eqnarray}\nonumber
&& [e,f]= h,\hspace{1mm} [h,f]=-2f,\hspace{1mm}  [h,e]=2e.
\end{eqnarray}

\begin{defn*}\label{SchFam}
    Under our assumptions of a fixed PLQS    $(\mathcal{V},q,w)$,   and a fixed $\kappa \in \R$, for every $\lambda\in \C$, we define a second-order differential operator on the cone
\[\boxed{S_{(q,w,\kappa)}(\lambda)=S_{\kappa}(\lambda):=\kappa -i\left(\lambda e+f \right)  .}  \]
We call $S_{(q,w,\kappa)}$  \textbf{the  Schr\"odinger family associated with} $(\mathcal{V},q,w,\kappa)$.
\end{defn*}
\begin{remark*}
 For every $\lambda\in \R$,
 $S_{(q,w,\kappa)}(\lambda)^*=S_{(q,w,\kappa)}(\lambda)$.
\end{remark*}

Similarly,  for any $\mathfrak{sl}_2$-triple $\{E,H,F\}$ in any Lie algebra $\fg$ we define a one-parameter family in $\mathcal{U}(\fg)$ (\textbf{the abstract  Schr\"odinger family}) by 
\[S^{Ab}_{\kappa}(\lambda):=  \kappa -i\left( \lambda  E+F\right), \quad \lambda\in \C.\]

\subsection{Motivation  for the definition  of  the Schr\"odinger  family}\label{physmot} 
In this paragraph we compare our Schr\"{o}dinger family as a family of differential operators on the upper  half cone  with the Schr\"{o}dinger operator of the hydrogen atom in physics.

In Physics, the Schr\"{o}dinger operator    $H_{phys}$ on $\R^{3}\setminus \{0\}$ is given by 
\[{ H_{phys}=-\Delta-\frac{\kappa}{r} } \] 
where $\Delta$ is the Laplacian, $r$ the radius function and $\kappa$ is a positive  constant.  The  corresponding  time-independent  {Schr\"{o}dinger }    equation on $\R^{3}\setminus \{0\}$ is given by 
\[{ H_{phys}\psi =E \psi , } \]
where $E$ is the energy (the eigenvalue of  $H_{phys}$).

Comparing the above with our  suggested model we can choose  the $4$-dimensional PLQS in which $\mathcal{V}=\R^4$, $q(x,y,z,w)=x^2+y^2+z^2-w^2$,   and the fixed linear functional in our  PLQS is the linear functional $w$ (determined by the coordinate $w$).

Projection from the 
upper half cone
 \[{C}_+(\R)=\{(x,y,z,w)\in C_0(\R)| w>0 \} \]
 onto $\R^3\setminus \{0\} $
is a diffeomorphism that enables us to express the Schr\"{o}dinger family  $S_{(q,w,\kappa)}(\lambda=E)$  in terms of  the Cartesian coordinates $x,y,z$, as 
\[ r\Delta  + \kappa +\lambda r=r\left(-H_{phys}+E\right).   \]
That is, on the space of smooth functions  on $\R^3\setminus\{0\}$,
the operator $S_{(q,w,\kappa)}(E)$   is proportional to   $(H_{phys}-E)$ and both have the same kernel. 

\section{Spectrum of the Schr\"odinger family}\label{se6} 
In this section we discuss the  spectrum of a family of operators on a topological  vector space. 

We compare   the  spectrum of the Schr\"{o}dinger family in $\mathcal{S}(\boldsymbol{H})$  with the spectrum of the Schr\"{o}dinger operator of the hydrogen atom in physics. We show that the latter coincides with the spectrum of the Schr\"{o}dinger family in $\mathcal{S}(\boldsymbol{H})_+$ the subspace   consisting of  functions that are supported on the "positive half cone".
\subsection{Invertibility and spectrum}
In this subsection we define the spectrum of a family of continuous operators on a topological vector space. 

Our goal is to calculate the spectrum of the Schr\"{o}dinger family in $\mathcal{S}(\boldsymbol{H})$. In this setting the  Schr\"{o}dinger family acts as a family of operators  within a representation of $SL_2(\R)$ that commute with $SO(n)$.  We show that it is enough to calculate the spectrum within each $SO(n)$-isotypic component.
 
\subsubsection{The spectrum of a family of operators}
In functional analysis, given  a continuous operator $T$ on a dense linear subspace $\boldsymbol{W}$ of a topological vector space $\boldsymbol{H}$ (typically $\boldsymbol{H}$ is a Hilbert space) the spectrum of $T$ is defined as  the collection of all  $\lambda\in \C$ for which the operator $T-\lambda :\boldsymbol{W}\longrightarrow \boldsymbol{H}$  does not have a continuous inverse, see e.g., Section VI3 of \cite{reed}. In physics, given an operator $T$ as above, one is usually   interested in a spectral decomposition of $\boldsymbol{H}$ with respect to $T$. Sometimes there are not enough eigenvectors in $\boldsymbol{H}$ for a spectral decomposition and a natural extension is the dual space. This motivates the following definition.

Let $\boldsymbol{W}$ be a complex topological vector space. We denote by $\boldsymbol{W}'$ the continuous complex-linear dual of $\boldsymbol{W}$ and define the \textit{ Hermitian dual} of $\boldsymbol{W}$ by
\[
\boldsymbol{W}^{\dagger}:=\overline{\boldsymbol{W}'}.
\]
Equivalently, $\boldsymbol{W}^{\dagger}$ is the space of continuous conjugate-linear functionals on $\boldsymbol{W}$. If $T:\boldsymbol{W}\longrightarrow \boldsymbol{W}$ is continuous, we write
\[
T^{\dagger}:=\overline{T'}:\boldsymbol{W}^{\dagger}\longrightarrow \boldsymbol{W}^{\dagger}
\]
for the Hermitian adjoint of $T$, where $T'$ is the transpose on $\boldsymbol{W}'$. In the realization of $\boldsymbol{W}^{\dagger}$ as conjugate-linear functionals this means
\[
(T^{\dagger}\eta)(w)=\eta(Tw),\qquad \eta\in \boldsymbol{W}^{\dagger},\quad w\in \boldsymbol{W}.
\]

Let $\boldsymbol{W}$ be a topological vector space. Consider a family $\mathcal{F}$ of  continuous  operators 
on $\boldsymbol{W}$  parametrized by  a set $\Lambda$ (i.e., a map $\mathcal{F}:\Lambda \longrightarrow  \operatorname{End}(\boldsymbol{W})$).

\begin{defn*}
    Let  $\mathcal{F}:\Lambda \longrightarrow  \operatorname{End}(\boldsymbol{W})$ be a family of continuous linear operators.  For every $\lambda\in \Lambda$ we define the $\lambda$-solution space (of $\mathcal{F}$) in $\boldsymbol{W}$ to be the kernel of the Hermitian adjoint $\mathcal{F}(\lambda)^{\dagger}:\boldsymbol{W}^{\dagger}\longrightarrow \boldsymbol{W}^{\dagger}$. We denote this space by 
    $\operatorname{Ker}(\mathcal{F}(\lambda);\boldsymbol{W}^{\dagger})$. 

    We  define $\operatorname{Spec}_{\boldsymbol{W}}\{\mathcal{F}\},$  the spectrum of the family  $\mathcal{F}$, to be the collection of all $\lambda\in \Lambda$ with nonzero solution space. That is,    
\[\operatorname{Spec}_{\boldsymbol{W}}\{\mathcal{F}\}=\{\lambda \in \Lambda|\operatorname{Ker}(\mathcal{F}(\lambda);\boldsymbol{W}^{\dagger})\neq 0 \}.\]
 
\end{defn*}

\begin{example*} 
For a given continuous operator $T$ on a topological vector space $\boldsymbol{W}$, we obtain a family  $\mathcal{F}:\C\longrightarrow \operatorname{End}(\boldsymbol{W})$ via $\mathcal{F}(\lambda)=T-\lambda$. Then
\[
\operatorname{Ker}(\mathcal{F}(\lambda);\boldsymbol{W}^{\dagger})
=\{\eta\in \boldsymbol{W}^{\dagger}\mid T^{\dagger}\eta=\overline{\lambda}\eta\}.
\]
Thus $\operatorname{Spec}_{\boldsymbol{W}}\{\mathcal{F}\}$ is the complex conjugate of the point spectrum of $T^{\dagger}$.

When 
$\boldsymbol{W}=C^{\infty}_c(\R)$,
the space of  smooth compactly supported functions on the line equipped with the standard topology  and $T=\frac{d}{dx}$, one has $\operatorname{Spec}_{\boldsymbol{W}}\{\mathcal{F}\}=\C$. In this case $\operatorname{Ker}(\mathcal{F}(\lambda);\boldsymbol{W}^{\dagger})$ is spanned by the element of
\[
\mathcal{D}^{\dagger}(\R):=(C_c^{\infty}(\R))^{\dagger}=\overline{\mathcal{D}'(\R)}
\]
represented by the locally integrable function $u(x)=e^{-\overline{\lambda}x}$, namely by the conjugate-linear functional
\[
f\longmapsto \int_{\R}\overline{f(x)}e^{-\overline{\lambda}x}dx.
\]
Finally, when $\boldsymbol{W}$ is the space of Schwartz class functions on $\R$ and $T=\frac{d}{dx}$, the spectrum reduces to $i\R$.
\end{example*}

\subsubsection{The spectrum of the Schr\"{o}dinger family   within a representation of $SL_2(\R)$ }\label{333}
For any smooth representation  $(\pi,\boldsymbol{W})$ of $SL_2(\R)$,  $\pi(S^{Ab}_{\kappa})$ is a family of operators in $\operatorname{End}(\boldsymbol{W})$ parameterized by $\C$.

The following  lemma will be needed later on. We omit its simple proof. 
\begin{lemma*}\label{l111}
Let  $(\pi,\boldsymbol{W})$ be a smooth representation  of $SL_2(\R)$ and assume that $\kappa>0$.
Then 
\[\operatorname{Spec}_{\boldsymbol{W}}\{\pi(S^{Ab}_{\kappa})\}=\kappa^2\left(\operatorname{Spec}_{\boldsymbol{W}}\{\pi(S^{Ab}_{1})\}\right).\]
\end{lemma*}
\subsubsection{The spectrum of the Schr\"{o}dinger family  within a representation of $SL_2(\R) \times SO(n)$ }\label{334}
We will be interested in the spectrum of the Schr\"{o}dinger family within a representation of $SL_2(\R)\times SO(n)$, mainly in $\mathcal{S}(\boldsymbol{H})$. We show that in our setting calculating  the spectrum and solutions within each $SO(n)$-isotypic component is enough to  determine the spectrum and solutions in the ambient topological vector space. 

We shall use the following Peter--Weyl reduction, which applies to the smooth Fr\'{e}chet representations that occur below.

Let $K=SO(n)$, and let $(\pi,\boldsymbol{W})$ be a smooth Fr\'{e}chet representation of $SL_2(\R)\times K$.  For $\tau\in \widehat K$, let $\chi_\tau$ be the character of $\tau$ and define the continuous projection
\[
P_\tau(v):=\dim(\tau)\int_K \overline{\chi_\tau(k)}\,\pi(k)v\,dk,
\]
where $dk$ is the normalized Haar measure on $K$.  We denote the image of $P_\tau$ by $\boldsymbol{W}_\tau$.  Thus $\boldsymbol{W}_\tau$ is the $K$-isotypic component of type $\tau$.  By the Peter--Weyl theorem for Fr\'{e}chet representations of compact groups,
\[
\boldsymbol{W}_{fin}:=\bigoplus_{\tau\in\widehat K}\boldsymbol{W}_\tau
\]
is dense in $\boldsymbol{W}$.

As in the previous subsection, we suppress $\pi$ from the notation and write $S^{Ab}_{\kappa}(\lambda)$ for the corresponding continuous endomorphism of $\boldsymbol{W}$, and also for its restriction to every $\boldsymbol{W}_\tau$.

\begin{claim*}
Let $(\pi,\boldsymbol{W})$ be as above.  Then:
\begin{enumerate}
    \item For every $\lambda\in\C$ and every $\tau\in\widehat K$, restriction to $\boldsymbol{W}_{\tau}$ induces a canonical isomorphism
\[
\left(\operatorname{Ker}(S^{Ab}_{\kappa}(\lambda);\boldsymbol{W}^{\dagger})\right)_\tau
\simeq
\operatorname{Ker}(S^{Ab}_{\kappa}(\lambda);\boldsymbol{W}_{\tau}^{\dagger}).
\]
The inverse map is given by
\[
\eta\longmapsto \eta\circ P_{\tau}.
\]
    \item The spectrum of the Schr\"{o}dinger family in $\boldsymbol{W}$ is the union of the spectra on the $K$-isotypic pieces:
\[
\operatorname{Spec}_{\boldsymbol{W}}\{S^{Ab}_{\kappa}\}
=
\bigcup_{\tau\in\widehat K}
\operatorname{Spec}_{\boldsymbol{W}_\tau}\{S^{Ab}_{\kappa}\}.
\]
\end{enumerate}
\end{claim*}

\begin{proof}
Put $T_\lambda=S^{Ab}_{\kappa}(\lambda)$. Since $T_\lambda$ belongs to the image of $\mathcal{U}(\mathfrak{sl}_2)$, it commutes with the action of $K$. Hence it commutes with every $P_\sigma$, and each $\boldsymbol{W}_\sigma$ is $T_\lambda$-stable.

We first prove the isomorphism in (1).  The restriction map from the left-hand side to the right-hand side is well-defined because $\boldsymbol{W}_{\tau}$ is $T_\lambda$-stable.  If $\varphi\in (\boldsymbol{W}^{\dagger})_\tau$, then by Schur orthogonality for the Hermitian dual $\varphi$ vanishes on $\boldsymbol{W}_\sigma$ for every $\sigma\not\simeq\tau$.  Therefore, if its restriction to $\boldsymbol{W}_{\tau}$ is zero, then $\varphi$ vanishes on $\boldsymbol{W}_{fin}$, and hence $\varphi=0$ by density.  This proves injectivity.

Conversely, let
\[
\eta\in \operatorname{Ker}(S^{Ab}_{\kappa}(\lambda);\boldsymbol{W}_{\tau}^{\dagger}).
\]
Then $\widetilde{\eta}:=\eta\circ P_{\tau}$ is a continuous conjugate-linear functional on $\boldsymbol{W}$, and it belongs to the $\tau$-isotypic component of $\boldsymbol{W}^{\dagger}$.  Since $P_{\tau}$ commutes with $T_\lambda$, we have $T_\lambda^{\dagger}\widetilde{\eta}=0$.  Thus $\widetilde{\eta}$ belongs to the left-hand side and restricts to $\eta$ on $\boldsymbol{W}_{\tau}$.  This proves (1).

We now prove (2). If $\lambda\in\operatorname{Spec}_{\boldsymbol{W}_\tau}\{S^{Ab}_{\kappa}\}$, then a nonzero solution on $\boldsymbol{W}_\tau$ gives, by composition with $P_\tau$, a nonzero solution on $\boldsymbol{W}$. Hence the right-hand side is contained in the left-hand side.

For the opposite inclusion, let $0\neq\varphi\in\operatorname{Ker}(T_\lambda^{\dagger};\boldsymbol{W}^{\dagger})$.  Since $\boldsymbol{W}_{fin}$ is dense in $\boldsymbol{W}$, the functional $\varphi$ is nonzero on some $\boldsymbol{W}_\tau$.  Its restriction to this $T_\lambda$-stable subspace is a nonzero element of
\[
\operatorname{Ker}(S^{Ab}_{\kappa}(\lambda);\boldsymbol{W}_\tau^{\dagger}).
\]
Thus $\lambda\in\operatorname{Spec}_{\boldsymbol{W}_\tau}\{S^{Ab}_{\kappa}\}$ for some $\tau\in\widehat K$, which proves the reverse inclusion.
\end{proof}

 \subsection{The spectrum of  the Schr\"odinger family in $\mathcal{S}(\boldsymbol{H})$}
In this section we calculate the spectrum of   the Schr\"odinger family in $\mathcal{S}(\boldsymbol{H})$. We show that  the spectrum of   the Schr\"odinger family in $\mathcal{S}(\boldsymbol{H})_{+} $ coincides with the spectrum that is calculated in physics in the space of functions (or distributions) in $L^2(\R^3)$ satisfying  certain  boundary conditions. 

We  fix $n=4$ to be in agreement with the physical system of the  hydrogen atom. The presented results hold with appropriate adjustments for any    $n\in 2\mathbb{N}$ such that  $n\geq 4$.
\subsubsection{The main result}\label{t721}
\begin{theorem*}\label{t151}
Let $(\mathcal{V},q,w)$ be a $4$-dimensional  PLQS, let $\kappa$ be a fixed positive number, and  let $\boldsymbol{H}$ be the corresponding canonical Hilbert space.
\begin{eqnarray}\nonumber
      && \operatorname{Spec}_{\mathcal{S}(\boldsymbol{H})_+}\{S^{Ab}_{\kappa} \}=  \left\{-\frac{\kappa^2}{(2n)^2}
|n\in \N\right\}\cup[0,\infty) \\ \nonumber
      && \operatorname{Spec}_{\mathcal{S}(\boldsymbol{H})_-}\{S^{Ab}_{\kappa} \}=  (0,\infty).
  \end{eqnarray}
\end{theorem*}

The representation $(\pi,G(\mathcal{V}), \boldsymbol{H})$ gives rise to a representation of $SL_2(\R)\times SO(3)$ on $\boldsymbol{H}$. 
We shall prove the above theorem using  the representation $(\pi,SL_2(\R)\times SO(3), \boldsymbol{H})$ and 
  Claim \ref{334} by reducing the  spectrum calculations to calculations within  each $SO(3)$-isotypic component. This is  stated precisely in   Proposition \ref{623} below.  
  Each  $SO(3)$-isotypic component
carries a discrete series representation of $SL_2(\R)$. We explain how the above mentioned fact together with Casselman-Wallach  theorem reduces the problem of calculating the spectrum of the Schr\"odinger family  in $\mathcal{S}(\boldsymbol{H})$ to a problem  that can be handled completely within representation theory of $SL_2(\R)$.

\subsubsection{The $(SL_2(\R) \times SO(3))$ representation on 
$\boldsymbol{H}$ }\label{acti}

Let $(\pi,G(\mathcal{V}), \boldsymbol{H})$ be the   unitary irreducible representation  associated with the given  
PLQS $(\mathcal{V},q,w)$ that was  discussed  in Subsection \ref{dis}.

To simplify our discussion we will identify $\mathcal{V}$ with $\R^4$, the Lorentzian quadratic form $q$ with the standard Lorentzian quadratic form $q(x,y,z,w)=x^2+y^2+z^2-w^2$,  and the fixed linear functional $w$ with the coordinate  $w$.

We let ${K}_{w}=SO(3)$  be the connected component of the identity in the  stabilizer of $w$ in $O(q)=O(3,1)$. The connected component of the identity in the centralizer of  $SO(3)$ in $G(\mathcal{V})$ is isomorphic to 
$PSL_2(\R)$.
Hence there is a unique homomorphism $\xi:SL_2(\R)\times SO(3)\longrightarrow  G(\mathcal{V})$
such that; 
\begin{enumerate}
    \item Restricted to $SO(3)$, $\xi$ is the identity. 
    \item For the elements 
    \[E:=\begin{pmatrix}
0 & 1 \\
0 & 0 
\end{pmatrix}, H:=\begin{pmatrix}
1 & 0 \\
0 & -1 
\end{pmatrix}\]
  in   $\mathfrak{sl}_2(\R)$, as differential operators on the cone 
\[\pi(\xi(E))=iw,\quad \pi(\xi(H))=h=2\chi+2. \] 
\end{enumerate}

\begin{remark*}
It can be shown that $\pi(\xi(\mathfrak{sl}_2(\R)))=\mathfrak{l}_w\cap \fs^*(\R)$, 
$\pi(\xi(\mathfrak{so}(3)))=\fk_w\cap \fs^*(\R)$, where
$\mathfrak{l}_w$ and $ \fk_w$ are as in  Section \ref{PLQS}.  
Explicitly, 
\begin{eqnarray}\nonumber
 && \pi(\xi(\mathfrak{so}(3)))=\operatorname{Span}_{\R}\{x\partial_y-y\partial_x,y\partial_z-z\partial_y,z\partial_x-x\partial_z \} \\ \nonumber 
&&\pi(\xi(\mathfrak{sl}_2(\R)))=\operatorname{Span}_{\R}\{e:=iw,f:=i\Psi^Q(w),h=2\chi+2\}.
\end{eqnarray}
Moreover, Proposition \ref{a823} and Proposition \ref{PLQS} imply that  the pair 

\noindent $(\pi(\xi(\mathfrak{sl}_2(\R))),\pi(\xi(\mathfrak{so}(3))))$ is a  dual pair in $\fs^*(\R)$. 

\end{remark*}

\subsubsection{The $(SL_2(\R) \times SO(3))$-isotypic decomposition of
$\boldsymbol{H}_{\pm}$ }\label{iso}
For   $\ell\in \mathbb{Z}_{\geq 0}$ we denote 
  by $\mathcal{Y}^{\ell}$ the subspace  of $L^2(S^2)$ consisting  of spherical harmonics of degree $\ell$. Each $\mathcal{Y}^{\ell}$ is an irreducible representation of $SO(3)$, and up to equivalence there are no others.

 The  $SO(3)$-isotypic component of  
$\boldsymbol{H}$ corresponding to 
$\mathcal{Y}^{\ell}$
will be denoted by $\boldsymbol{H}_{\ell}$ and of $\mathcal{S}(\boldsymbol{H})$ by $\mathcal{S}(\boldsymbol{H})_{\ell}$.
 
We denote  by $D^{\pm}_{2\ell+2}$
 the irreducible unitary discrete series representation  of $SL_2(\R)$ 
 with 
 $SO(2)$-types that are given by $\pm\{2\ell+2,2\ell+4,...\}$, respectively. More concretely,  these  integers are the eigenvalues of  the operator 
 $H_c:=i\begin{pmatrix}
0 & -1 \\
1 & 0 
\end{pmatrix}$,  the generator for the compact one-parameter subgroup $SO(2)$ of  $SL_2(\R)$.  

We set $D_{2\ell+2}:=D^{+}_{2\ell+2}\oplus D^{-}_{2\ell+2}$.

The following claim about the action of $(SL_2(\R) \times SO(3))$  on 
$\mathcal{S}(\boldsymbol{H})$ can be verified directly.
\begin{claim*}
  On $\mathcal{S}(\boldsymbol{H})$ the actions of the Casimirs of $\mathfrak{sl}_2(\R)$ and $\mathfrak{so}(3)$ coincide.  
\end{claim*}

 The  claim above together with  the known classification of the unitary dual of $SL_2(\R)$ imply that for $\ell>0$
the only representations of $SL_2(\R)$ that can appear in  $\boldsymbol{H}_{\ell}$ are 
$D^{\pm}_{2\ell+2}$ and those that can appear  in $\boldsymbol{H}_{0}$ are 
$D^{\pm}_{2}$ and the trivial representation.  The next proposition states explicitly which representations of $SL_2(\R) \times SO(3)$ appear in $\boldsymbol{H}_{\ell}$. 

\begin{proposition*}
As representations of  $SL_2(\R) \times SO(3)$ on Hilbert spaces, for every $\ell\in \N_0$,
\begin{eqnarray}  
  &&\boldsymbol{H}_{\ell}\cong D_{2\ell+2}\otimes_{\C} \mathcal{Y}^{\ell},\\ \nonumber  
  &&(\boldsymbol{H}_{\pm})_{\ell}\cong D^{\pm}_{2\ell+2}\otimes_{\C} \mathcal{Y}^{\ell}
 \end{eqnarray}
\end{proposition*}
This follows from Sections 1.3 and 3.1   of \cite{MR2401813}.

\begin{remark*}
    The multiplicities of $D^{\pm}_{2\ell+2}\otimes_{\C} \mathcal{Y}^{\ell}$ in the canonical Hilbert space $\boldsymbol{H}$ can also be  deduced by calculations similar to those that are given below  in  the proof of Proposition \ref{623} in the case of $\lambda=0$.
\end{remark*}
\subsubsection{Generalized spherical coordinates on the 
cone}

The cone ${C}_0(\R)$ is a double cover of $\R^3\setminus\{0\}$, this motivates us  to use  new, better-suited coordinates on ${C}_0(\R)$.

 We introduce generalized spherical coordinates on the cone via 
 \begin{eqnarray}\nonumber
&&\R\setminus\{0\}\times S^{2}\longrightarrow {C}_0(\R) \\ \nonumber
&& (t,u)\longmapsto (tu,t).
\end{eqnarray}
  This leads to an isomorphism  
\begin{eqnarray} \nonumber
  &&\boldsymbol{H}\cong L^2\left(\R\setminus\{0\},\frac{1}{2}|t|dt \right)\otimes_{\C} L^2(S^{2}), 
 \end{eqnarray}
respecting the standard actions of $SO(3)$. From this we see that 
for any $\ell \in \N_0$, 
 \[\boldsymbol{H}_{\ell}\cong L^2\left(\R\setminus\{0\},\frac{1}{2}|t|dt \right)\otimes_{\C} \mathcal{Y}^{\ell}. \]

\begin{proposition*}
The action of $\mathfrak{sl}_2(\R)$ on $\mathcal{S}(\boldsymbol{H})_{\ell}$     in terms of generalized spherical coordinates  is given by 
 \begin{eqnarray}\nonumber
&&\pi_{\ell} (E)=  i{t}, \\ \nonumber
&& \pi_{\ell}(H)  = 2 t\partial_t +2  ,\\ \nonumber
&& \pi_{\ell}(F)=   i\left(t\partial_t^2+2\partial_t-\frac{\ell(\ell+1)}{t} \right).  \end{eqnarray}
\end{proposition*} 

The proof follows from direct calculations using claim \ref{iso}.

\subsubsection{Implications of the Casselman-Wallach  theorem}\label{Imp}
In this paragraph  we explain how  Casselman-Wallach globalization theorem allows us to do all the relevant calculations for determining the spectrum of the Schr\"{o}dinger family completely within representation theory of $SL_2(\R)$. However, in this paper we solve this problem working with  these representations of $SL_2(\R)$ as they concretely appear inside the Hilbert  space $\boldsymbol{H}$ of functions on the cone. 

 By Casselman-Wallach globalization theorem \cite{Wallach3,Casselman,Wallach2},  any two 
admissible  smooth moderate growth Fr\'{e}chet representations of  the same   reductive Lie group  that are infinitesimally  equivalent  must be  isomorphic as topological representations. 

The topological vector space $\mathcal{S}(\boldsymbol{H})$
 is an admissible smooth Fr\'{e}chet representation of moderate growth of $G\simeq O(4,2)$. 
The subspaces $\mathcal{S}(\boldsymbol{H}_{\pm})_{\ell}$ 
are automatically smooth Fr\'{e}chet representations of moderate growth of $SL_2(\R) \times SO(3)$.  By  Proposition \ref{iso}  these representations are also admissible, and
\begin{eqnarray} \nonumber
  &&\mathcal{S}(\boldsymbol{H}_{+})_{\ell}\cong (D^{+}_{2\ell+2})^{\infty}\otimes_{\C} \mathcal{Y}^{\ell}, \quad \mathcal{S}(\boldsymbol{H}_{-})_{\ell}\cong (D^{-}_{2\ell+2})^{\infty}\otimes_{\C} \mathcal{Y}^{\ell}, 
 \end{eqnarray}
 where on the right-hand side, smoothness is with respect to the action of $SL_2(\R)$. In appendix \ref{Concrete} we explicitly describe  the isotypic components 
 $\mathcal{S}(\boldsymbol{H}_{\pm})_{\ell}$.
\subsubsection{The spectrum of  the Schr\"odinger family in $\mathcal{S}(\boldsymbol{H}_{\pm})$ }\label{623}
\begin{proposition*}
Let $(\mathcal{V},q,w)$ be a $4$-dimensional  PLQS, let $\kappa$ be a fixed positive number, and let $\boldsymbol{H}$ be the corresponding canonical Hilbert space. Fix $\ell \in \Z_{\geq 0}$. 
\begin{enumerate}
    \item 
    \[
\operatorname{Spec}_{\mathcal{S}(\boldsymbol{H}_{+})_{\ell}}\{S_{\kappa}\}=\left\{-\frac{\kappa^2}{(2n)^2}\mid \ell<n\in \N \right\}\cup [0,\infty).
\]
\[
\operatorname{Spec}_{\mathcal{S}(\boldsymbol{H}_{-})_{\ell}}\{S_{\kappa}\}=(0,\infty).
\]
  \item For $\lambda\in \C$, the kernel of the operator $S_{\kappa}(\lambda)$ in $(\mathcal{S}(\boldsymbol{H}_{+})_{\ell})^{\dagger}$ is nonzero if and only if one of the following hold:
\begin{enumerate}
    \item $0\leq \lambda$.
    \item $\lambda=-\frac{\kappa^2}{(2n)^2}$ for some $n\in \N$ such that $\ell<n$.
\end{enumerate}
Moreover, in these cases as a representation of $SO(3)$ the kernel is isomorphic to $\mathcal{Y}^{\ell}$. 
\item For $\lambda\in  \C$, the kernel of the operator $S_{\kappa}(\lambda)$ in $(\mathcal{S}(\boldsymbol{H}_{-})_{\ell})^{\dagger}$ is nonzero if and only if $\lambda>0$.
Moreover, in these cases as a representation of $SO(3)$ the kernel is isomorphic to $\mathcal{Y}^{\ell}$.
\end{enumerate}
\end{proposition*}
The above proposition and Claim \ref{334} prove Theorem \ref{t721}.

\begin{proof}
In Proposition \ref{realspec} in the appendix we show that in any admissible smooth Fr\'echet representation of $SL_2(\R)$ the spectrum of the Schr\"odinger family is always contained in $\R$. In particular this holds for $\mathcal{S}(\boldsymbol{H}_{\pm})_{\ell}$.  Taking this into account and using Lemma \ref{333}, it is enough to determine for which $\lambda\in  \R$ the kernel of the operator $S:=S_1(\lambda)$ in the space $(\mathcal{S}(\boldsymbol{H}_{\pm})_{\ell})^{\dagger}$ is nonzero, and in these cases to calculate its dimension. 

We fix $\lambda\in \R$. Our strategy for determining the kernel of $S^{\dagger}$ is as follows.
We find another, simpler operator $T$ that is conjugate to $S$ by the action of some element $A\in SL_2(\R)$.
Then $A$ defines an isomorphism between the kernels of $T^{\dagger}$ and $S^{\dagger}$ on the Hermitian duals, so it is enough to compute the dimension of the kernel of $T^{\dagger}$.

We treat the three cases of negative $\lambda$, $\lambda=0$, and positive $\lambda$, separately. 

\textbf{The case of $\lambda<0$.} We set $\nu:=\sqrt{|\lambda|}$ a positive square root. In this case there exists an element $A\in SL_2(\R)$ such that $\operatorname{Ad}_A(S)=1+i\nu E-i\nu F$. Set $T=1+i\nu E-i\nu F$. As a result the kernel of the operator $S^{\dagger}$ in $(\mathcal{S}(\boldsymbol{H}_{\pm})_{\ell})^{\dagger}$ is isomorphic to the kernel of the operator $T^{\dagger}$ in the same space.

This last kernel is equal to the eigenspace of the operator $(i(F-E))^{\dagger}$ in $(\mathcal{S}(\boldsymbol{H}_{\pm})_{\ell})^{\dagger}$ corresponding to the eigenvalue $\nu^{-1} $. The one-parameter subgroup generated by  $F-E$ inside $SL_2(\R)$ is $SO(2)$.  The  eigenvalues of $F-E$ in any  representation of $SL_2(\R)$ belong to $i\mathbb{Z}$. Moreover the eigenvectors are automatically smooth vectors; hence we can solve for the kernel in $\mathcal{S}(\boldsymbol{H}_{\pm})_{\ell}$ instead of in $(\mathcal{S}(\boldsymbol{H}_{\pm})_{\ell})^{\dagger}$. 

It follows from  Proposition   \ref{iso}, that the eigenvalues of $i(F-E)$ in $\mathcal{S}(\boldsymbol{H}_{\pm})_{\ell}$ are $\pm 2 (k+\ell+1)$ with $k\in \mathbb{N}_0$, and that each eigenspace is isomorphic as a representation of $SO(3)$ to the corresponding irreducible space of spherical harmonics.  

Hence for any $\lambda<0$,  
\begin{enumerate}
    \item The kernel of $S^{\dagger}$ in $(\mathcal{S}(\boldsymbol{H}_{-})_{\ell})^{\dagger}$ is zero for any $\ell\in \Z_{\geq0}$.
    \item The kernel of $S^{\dagger}$ in $(\mathcal{S}(\boldsymbol{H}_{+})_{\ell})^{\dagger}$ is nonzero if and only if $\lambda=-\frac{1}{(2n)^2}$ for some $n\in \N$ such that $\ell<n$.
 Moreover, for $\lambda=-\frac{1}{(2n)^2}$ with $\ell<n$, the kernel is isomorphic to $\mathcal{Y}^{\ell}$.
\end{enumerate}
Before turning to the cases $\lambda=0$ and $\lambda>0$, we record how the
restriction to the open half-cones will be used. Set
\[
\mathcal{D}^{\dagger}(C_{\pm}(\R))
:=
(C_c^{\infty}(C_{\pm}(\R)))^{\dagger}
=
\overline{\mathcal{D}'(C_{\pm}(\R))}.
\]
Restriction of test functions gives a natural map
\[
r_{\pm}:
(\mathcal{S}(\boldsymbol{H}_{\pm})_{\ell})^{\dagger}
\longrightarrow
\mathcal{D}^{\dagger}(C_{\pm}(\R))_{\ell}.
\]
This map is not injective in general. In the generalized spherical
coordinates of Subsection \ref{iso}, together with the concrete realization
of Proposition \ref{Concrete}, restricting to the open half-cone forgets
precisely the possible contribution supported on the boundary
\[
\{0\}\times S^2
\]
of the corresponding closed half-cone. We shall use the standard structure
theorem for distributions supported on a submanifold: in local coordinates
\((t,u)\), where the submanifold is given by \(t=0\), such a distribution is
a finite sum of derivatives of \(\delta_0\) in the transverse variable \(t\),
with distributional coefficients in the remaining variables \(u\); see, for
example, \cite[Section~2.3]{HormanderLPDE1} or
\cite[Chapter~3]{FriedlanderJoshi}. Passing to complex conjugates gives
the corresponding statement for the spaces \(\mathcal{D}^{\dagger}\).

In the present coordinates this means the following. Write an element of
\(\mathcal{S}(\boldsymbol{H}_{+})_{\ell}\) as
\[
t^{\ell}a(t)Y(u),
\qquad
a\in\mathcal{S}([0,\infty)),\quad Y\in\mathcal{Y}^{\ell},
\]
and similarly write an element of \(\mathcal{S}(\boldsymbol{H}_{-})_{\ell}\)
with \(a\in\mathcal{S}((-\infty,0])\). Then every element of
\(\operatorname{Ker}r_{\pm}\) is a finite sum of functionals of the form
\[
J^{\pm}_{j,\alpha}(t^{\ell}a(t)Y(u))
:=
\overline{a^{(j)}(0)}\,\alpha(Y),
\qquad j\geq 0.
\]
Here \(\alpha\in(\mathcal{Y}^{\ell})^{\dagger}\). Equivalently, \(\alpha\)
is the restriction to the \(\ell\)-th spherical harmonic component of a
conjugate-linear distribution on \(S^2\). For the negative half-line,
\(a^{(j)}(0)\) denotes the left derivative at \(0\).

Thus, when \(T\) is the conjugate of \(S\) used below, it is enough to check
that
\[
\operatorname{Ker}(T^{\dagger})\cap \operatorname{Ker}r_{\pm}=0.
\]
After this check, the restriction map is injective on
\[
\operatorname{Ker}\bigl(T^{\dagger};
(\mathcal{S}(\boldsymbol{H}_{\pm})_{\ell})^{\dagger}\bigr),
\]
and the kernel may be computed on the open half-cone.

\textbf{The case of $\lambda=0$.}
In this case there exists an element  $A\in SL_2(\R)$ such that $\operatorname{Ad}_A(S)=1+iE$. Set $T=1+iE$. As a result the kernel of the operator $S^{\dagger}$ in $(\mathcal{S}(\boldsymbol{H}_{\pm})_{\ell})^{\dagger}$ is isomorphic to the kernel of the operator $T^{\dagger}$ in the same space.

We first check injectivity of the restriction map on this kernel. In generalized spherical coordinates $T$ acts by multiplication by $1-t$. Hence
\[
T^{\dagger}J^{\pm}_{j,\alpha}=J^{\pm}_{j,\alpha}-jJ^{\pm}_{j-1,\alpha},\hspace{2mm}  \text{where} \hspace{2mm} J^{\pm}_{-1,\alpha}:=0.
\]
On the finite-dimensional space spanned by $J^{\pm}_{0,\alpha},\ldots,J^{\pm}_{N,\alpha}$ this operator is triangular with diagonal entries equal to $1$. Hence
\[
\operatorname{Ker}(T^{\dagger};\operatorname{Ker}r_{\pm})=0.
\]
Therefore $r_{\pm}$ is injective on $\operatorname{Ker}(T^{\dagger};(\mathcal{S}(\boldsymbol{H}_{\pm})_{\ell})^{\dagger})$.

On the open half-cone, $T^{\dagger}$ is given by multiplication by $1-t$:
\[
\left(\pi_{\ell,\pm}(T)\right)^{\dagger}(\varphi)=(1-t)\varphi, \quad \forall \varphi\in \mathcal{D}^{\dagger}(C_{\pm}(\R))_{\ell}.
\]
Hence its kernel on $\mathcal{D}^{\dagger}(C_{+}(\R))_{\ell}$ is spanned by the elements $\delta(t-1)\otimes \alpha$, where $\alpha\in(\mathcal{Y}^{\ell})^{\dagger}$. Using the $SO(3)$-invariant Hermitian inner product to identify $(\mathcal{Y}^{\ell})^{\dagger}$ with $\mathcal{Y}^{\ell}$, this kernel is isomorphic to $\mathcal{Y}^{\ell}$. These elements extend continuously to $(\mathcal{S}(\boldsymbol{H}_{+})_{\ell})^{\dagger}$, and the injectivity just proved shows that they exhaust the kernel of $T^{\dagger}$ there. Thus the corresponding kernel of $S^{\dagger}$ is isomorphic to $\mathcal{Y}^{\ell}$.

On the other hand the kernel of $T^{\dagger}$ in $\mathcal{D}^{\dagger}(C_{-}(\R))_{\ell}$ is trivial and, by the same injectivity on the kernel, so is the kernel of $S^{\dagger}$ in $(\mathcal{S}(\boldsymbol{H}_{-})_{\ell})^{\dagger}$. \\

\textbf{The case of $\lambda>0$.}  
We set  $\nu:=\sqrt{\lambda}$ a positive square root. 
In this case there exists an element  $A\in SL_2(\R)$ such that $\operatorname{Ad}_A(S)=1+i\nu  H$. Set $T=1+i\nu  H$.  As a result the kernel of $S^{\dagger}$ in $(\mathcal{S}(\boldsymbol{H}_{\pm})_{\ell})^{\dagger}$ is isomorphic to the kernel of $T^{\dagger}$ in the same space.

We again first check injectivity of the restriction map on the relevant kernel. Under the identification of Proposition \ref{Concrete}, write an element of $\mathcal{S}(\boldsymbol{H}_{\pm})_{\ell}$ as $t^{\ell}a(t)Y(u)$, with $a$ in the corresponding Schwartz space on the closed half-line.  On this space,
\[
T(t^{\ell}a(t)Y(u))
=t^{\ell}\left( (1+2i\nu(\ell+1))a(t)+2i\nu t a'(t)\right)Y(u).
\]
Therefore
\[
T^{\dagger}J^{\pm}_{j,\alpha}=\left(1-2i\nu(\ell+1+j)\right)J^{\pm}_{j,\alpha}.
\]
The scalar $1-2i\nu(\ell+1+j)$ is never zero. Consequently
\[
\operatorname{Ker}(T^{\dagger};\operatorname{Ker}r_{\pm})=0,
\]
and the restriction map $r_{\pm}$ is injective on $\operatorname{Ker}(T^{\dagger};(\mathcal{S}(\boldsymbol{H}_{\pm})_{\ell})^{\dagger})$.

In terms of the above-mentioned  generalized spherical coordinates, the equation on the open half-cone is
\[
\left(\pi_{\ell,\pm}(T)\right)^{\dagger}(\varphi)=\left(1+2i\nu +  2i\nu t\partial_t \right)\varphi=0,
\qquad \varphi\in \mathcal{D}^{\dagger}(C_{\pm}(\R))_{\ell}.
\]
Hence the  solution space in $\mathcal{D}^{\dagger}(C_{+}(\R))_{\ell}$ consists of elements represented by the locally integrable functions $t^{-1+\frac{i}{2\nu }}Y^{\ell}(u)$, for $Y^{\ell}\in \mathcal{Y}^{\ell}$, and the solution space in $\mathcal{D}^{\dagger}(C_{-}(\R))_{\ell}$ consists of the corresponding elements represented by $|t|^{-1+\frac{i}{2\nu }}Y^{\ell}(u)$.  Explicitly, on $C_{+}(\R)$, using generalized spherical coordinates and the  measure on the cone, such an element is given by 
\[
f\in C^{\infty}_c((0,\infty)\times S^2) \longmapsto \langle \varphi,f \rangle :=\frac{1}{2}\int_{0}^{\infty}\int_{S^2}\overline{f(t,u)}t^{-1+\frac{i}{2\nu }}Y^{\ell}(u)t\,dt\,du,
\]
and on $C_{-}(\R)$ by 
\[
f\in C^{\infty}_c((-\infty,0)\times S^2) \longmapsto \langle \varphi,f \rangle :=\frac{1}{2}\int^{0}_{-\infty}\int_{S^2}\overline{f(t,u)}|t|^{-1+\frac{i}{2\nu }}Y^{\ell}(u)|t|\,dt\,du.
\]
The explicit  description of the $SO(3)$-isotypic pieces of $\mathcal{S}(\boldsymbol{H}_{\pm})$   given in Section  \ref{Concrete} implies that these elements extend to continuous conjugate-linear functionals on $\mathcal{S}(\boldsymbol{H}_{\pm})_{\ell}$. Since the restriction map is injective on the kernel, these extensions exhaust $\operatorname{Ker}(T^{\dagger};(\mathcal{S}(\boldsymbol{H}_{\pm})_{\ell})^{\dagger})$. Therefore for $\lambda>0$ the kernel is isomorphic to $\mathcal{Y}^{\ell}$ for both signs.
\end{proof}

\begin{remark*}
For $\lambda<0$, the preceding proof gives a conceptual explanation of
the discrete energy formula
\[
E_n=-\frac{\kappa^2}{(2n)^2}.
\]
Indeed, the factor $2n$ is an $SO(2)$-weight in a  representation of $SL_2(\R)$.
\end{remark*}

\section{Appendix }\label{Ap3}
\subsection{The commutation relations in $\fs$}\label{crs}
Recall that the Lie algebra $\tilde{\fs}$ that was introduced in Section \ref{222} is the direct sum of 
\begin{eqnarray}\nonumber
&&\mathcal{D}^{0,1}(C)\cong \mathcal{D}^{0,1}(V)=\{\tau(v)|v\in V \}, \\ \nonumber
&&\mathcal{D}^{1,0}(C)\cong \mathcal{D}^{1,0}(V)_I=\operatorname{span}_{\C}\{L_{u,v}|u,v\in V \}\oplus \C h \oplus \C, \\ \nonumber
&&\mathcal{D}^{2,-1}(C)\cong  \mathcal{D}^{2,-1}(V)_I=\{\Psi^Q(\tau(v))|v\in V\}.
\end{eqnarray}

 The subalgebras $\mathcal{D}^{0,1}(V)$ and $\mathcal{D}^{2,-1}(V)_I$ are commutative and the element $1$ is central.   By direct calculations, for any $u,v,w,z\in V$,
\begin{eqnarray}\nonumber
&&[L_{u,v},L_{w,z}]=B(v,w)L_{u,z}-B(u,w)L_{v,z}+B(u,z)L_{v,w}-B(v,z)L_{u,w}\\ \nonumber
&&[L_{u,v},\tau(w)]=B(v,w)\tau(u)-B(u,w)\tau(v)\\ \nonumber
&&[L_{u,v},\Psi^Q(\tau(w))]=B(v,w)\Psi^Q(\tau(u))-B(u,w)\Psi^Q(\tau(v)) \\ \nonumber
&& [\Psi^Q(\tau(v)),\tau(u)]=2L_{v,u}-B(v,u)h\\ \nonumber
&& [h,\Psi^Q(\tau(v))]=-2\Psi^Q(\tau(v))\\ \nonumber
&& [h,\tau(u)]=2\tau(u)\\ \nonumber
&& [h,L_{u,v} ]=0.
\end{eqnarray}

\subsection{The Lie algebra $\fs$ as an  orthogonal Lie algebra}
 \subsubsection{The Lie algebra $\fs$ as an orthogonal  Lie algebra}\label{sss232}
 We let  $(V_2,Q_2)$ be the complexification of the  standard hyperbolic plane $(\R^{1,1}=\R^2,q_{1,1})$ from Subsection \ref{dis} with its basis $\{v_1,v_2\}$  of isotropic vectors satisfying $B_{1,1}(v_1,v_2)=1$.

Let $(V,Q)$ be a nondegenerate $n$-dimensional complex quadratic space. Similar to the construction given in Subsection \ref{dis},  here, in the complex algebraic setting,  we  set $\widetilde{V}=V\oplus V_2$, $\widetilde{Q}=Q \oplus Q_2$ and let $\widetilde{B}$  be the corresponding  symmetric bilinear form.
We let $G=G(V)=O(\widetilde{V},\widetilde{Q})$ be the orthogonal group    of $(\widetilde{V},\widetilde{Q})$,  $P$ be   
the stabilizer of $v_1$ in $G$ and  $N$ the unipotent radical of $P$.  We denote  the canonical copy of $O(V,Q)$ in $P$ by  $L$. 
We  let $A$ be the subgroup of $G$ consisting of all  the elements   that act on $V$ as the identity operator and preserves the line $\C v_1\subset V_2\subset \widetilde{V}$. The group $A$ is canonically isomorphic to $O_0(V_2,Q_2)$ (the connected component of the identity in $O(V_2,Q_2)$).

We denote the   Lie algebras of $G$, $P$, $N$, $L$ and $A$ by $\mathfrak{g}$, $\mathfrak{p}$, $\mathfrak{n}$, $\mathfrak{l}$ and $\mathfrak{a}$ respectively.

In this section we construct a canonical extension of the natural morphism of Lie algebras $\mathfrak{l} =\mathfrak{o}(Q) \longrightarrow \mathcal{D}(C)$ to a morphism of  Lie algebras  
$\psi: \mathfrak{g}\longrightarrow \mathcal{D}(C)$.
This will give 
us an epimorphism of associative algebras  
$\psi: \mathcal{U}(\mathfrak{g})\longrightarrow \mathcal{D}(C)$.

 We construct the  embedding in two steps. First we construct a morphism  $\mathfrak{p}\longrightarrow \mathcal{D}(C)$ and then we extend it to $\fg\longrightarrow \mathcal{D}(C)$.

The action of $A$ on  $v_1\in V_2$
defines an isomorphism $\rho:A\longrightarrow\C^{\times}$  given by
\[av_1= \rho(a)v_1, \quad \forall a\in A.\]
The  group $G$ acts by conjugation on its   Lie algebra $\fg$,  and using the homomorphism 
\[ \C^{\times}\xrightarrow{ \rho^{-1}}
A    \longhookrightarrow  G,\]
 $\C^{\times}$ also acts on $\fg$.

There is an obvious 
inclusion of Lie algebras $\iota_{\mathfrak{o}}:\mathfrak{l} \longrightarrow \fg$. There is also  a natural embedding of Lie algebras 
$\iota_{\mathcal{D}}:\mathfrak{l}  \longrightarrow \mathcal{D}$, obtained by differentiating the action of $L$ on regular functions on $V$. 

We have a natural isomorphism of abelian Lie algebras 
\begin{eqnarray}\nonumber
    &&f_{v_2}:\mathfrak{n}\longrightarrow V^*\subset \mathcal{D}(C) 
\end{eqnarray} 
given by 
\[\left(f_{v_2}(n)\right)(v):=B(nv_2,v) ,\]
for any $n\in \mathfrak{n}$ and $v\in {V}$, and 
where $nv_2$ stands for the application of the linear transformation $n$ to $v_2$.

Hence we obtain a unique  linear map $\widetilde{\psi}:\mathfrak{p} \longrightarrow \mathcal{D}(C)$ satisfying 
$\widetilde{\psi}\circ \iota_{\mathfrak{o}}=\iota_{\mathcal{D}}$ and $\widetilde{\psi}|_{\mathfrak{n}}=f_{v_2}$. It is easily verified that $\widetilde{\psi}$ is an $(L\times \C^{\times})$-equivariant embedding of complex Lie algebras. 

\begin{proposition*}
There exists a unique $ \C^{\times}$-equivariant morphism of complex Lie algebras $\psi:\fg\longrightarrow  \mathcal{D}(C) $ extending $\widetilde{\psi}$.
In addition $\psi$ is an embedding, its image is equal to $\fs$, and the induced map $\mathcal{U}(\fg)\longrightarrow  \mathcal{D}(C)$ is an epimorphism. 
\end{proposition*}
The proof is left to the reader.

\subsubsection{The induced real  structure on $\fs$ }\label{s2.5}

The construction of $\tilde{\fs}$ 
 out of a  complex quadratic space $(V,Q)$  can be repeated  over the field of real numbers. Hence, given a real form $(\mathcal{V},q)$ of $(V,Q)$, not necessarily Lorentzian,  there is a canonical real form $\fs(\R)$ of  $\fs$ that is compatible with $(\mathcal{V},q)$ and respects the decomposition $\tilde{\fs}=\mathcal{D}^{0,1}(C)\oplus \mathcal{D}^{1,0}(C)  \oplus   \mathcal{D}^{2,-1}(C)$.
Explicitly, it is given by the direct sum of 
\begin{eqnarray}\nonumber
&&  \mathcal{D}^{0,1}(V)(\R)=\{ \tau^Q(v)|v\in \mathcal{V} \}, \\ \nonumber
&&  \mathcal{D}^{1,0}(V)_I(\R)=\operatorname{span}_{\R}\{L_{u,v}|u,v\in \mathcal{V} \}\oplus \R h \oplus \R \\ \nonumber
&&   \mathcal{D}^{2,-1}(V)_I(\R)=\{\Psi^Q(\tau^Q(v))|v\in \mathcal{V}\}. 
\end{eqnarray}
Similar to proposition \ref{sss232},  we can show that $\fs(\R)\simeq \mathfrak{o}(\mathcal{V}\oplus \mathcal{V}_{1,1})$, where $\mathcal{V}_{1,1}$ is a two-dimensional Lorentzian quadratic space.

\subsubsection{The isomorphism $\fs(\R)\simeq \fs^*(\R)$ }\label{a823}
  In Section \ref{CmpRF} we introduced the real form $\fs(\R)$ (that was  further discussed in   in Section \ref{s2.5}).
The main real form of $\fs$ that is used in this paper is $\fs^*(\R)$ that was defined in Section \ref{canreal}.
 
In this section we build an explicit isomorphism between 
$\fs(\R)$ and $ \fs^*(\R)$.  

We define a linear invertible transformation   $\mathcal{F}$ on the algebra     $\mathcal{D}(C)$  by $\mathcal{F}(d) =\exp(\operatorname{ad}_{i\frac{\pi}{4}h})(d)$.
 On $\mathcal{D}^{k,\ell}(C)$,  $\mathcal{F}$ is given by multiplication by $e^{i\frac{\ell\pi}{2}}$. 
 
\begin{proposition*}\label{p432}
The restriction  $\mathcal{F}|_{\fs^{*}(\R)}$ is an isomorphism of real Lie algebras from 
$\fs^{*}(\R)$ onto $\fs(\R)$. Moreover, it is the unique $O(q)\times \R^{\times}$-equivariant  isomorphism  of real Lie algebras from 
$\fs^{*}(\R)$ onto $\fs(\R)$ that on $\mathcal{D}^{0,1}(C)\cap \fs^{*}(\R)$ is given by multiplication by $i$.
\end{proposition*}
This immediately follows from the commutation relations in Appendix \ref{crs}.

\subsection{A Lie-theoretic description of $\Psi^Q$  }\label{JM}
In Section  \ref{s224} we introduced the isomorphism $\Psi^Q: \mathcal{D}^{0,1}(C)\longrightarrow \mathcal{D}^{2,-1}(C)$. In this section we  give a Lie-theoretic characterization of it using the Jacobson-Morozov theorem.  

Recall that $h=2\chi+n-2\in \fs$. Starting with a $w\in V^*$, such that $Q^*(w)\neq 0$, 
the triple $\{e_w:=w,h,f_w:=Q^*(w)^{-1}\Psi^Q(w)\}$
is an $\mathfrak{sl}_2$-triple in $\fs$ satisfying the standard commutation relations 
\[[h,e_w]=2e_w,\quad [h,f_w]=-2f_w,\quad [e_w,f_w]=h. \]

The Jacobson-Morozov theorem (see e.g., \cite[Ch. X]{MR1920389})  for the finite-dimensional simple Lie algebra $\fs$ implies that for $w$ as above there is a unique $f_w\in \fs$ such that 
$\{e_w,h,f_w\}$
is an $\mathfrak{sl}_2$-triple in $\fs$ satisfying the standard commutation relations. This uniquely determines the morphism $\Psi^Q$.

\begin{remark*}
   The Lie algebra generated by the above triple $\{e_w,h,f_w\}$ is $\mathfrak{l}_{w}\simeq \mathfrak{sl}_2(\C)$ that was introduced in Section \ref{PLQS} as   part of a dual pair. 
   
  We remark  that the   one-dimensional subspace spanned by $h$ is given by the intersection of all such $\mathfrak{l}_{w}$.    Also note that  $h$  is the unique element in $\C h\subset \mathfrak{l}_{w}$ satisfying $[h,e_w]=2e_w$.
\end{remark*}

\subsection{On the reductive dual pair symmetry of  $\fs$ }\label{l}
Here we give a proof of Proposition \ref{PLQS}.

\begin{proof}
Using the  explicit commutation relations in ${\fs}$, which were given in Appendix  \ref{crs}, it can be shown that
$\mathfrak{l}_{w}=\operatorname{span}_{\C}(\Psi^Q\left({w}\right), {w},h)\cong \mathfrak{sl}_2(\C)$ and $\mathfrak{k}_{w}=\operatorname{span}_{\C}\{ L_{u,v}|u,v\in V_w^{\perp}\}\simeq \mathfrak{o}(n-1,\C)$, where 
$V_w^\perp:=\ker w
 =\{v\in V\mid w(v)=0\}
 =\{v\in V\mid B(v_w,v)=0\}$. These facts together with the explicit commutation relations finish the proof of 1. The proof of 2. is similar. 
\end{proof}

\subsection{Concrete realization of $\mathcal{S}(\boldsymbol{H}_{\pm})_{\ell}$
  }\label{Concrete}
As explained in section  \ref{iso}, the unitary representation  $\pi_{\ell}$ of $SL_2(\R)\times SO(3)$ on $\boldsymbol{H}_{\ell}$ can be realized on $L^2\left(\R\setminus\{0\},\frac{1}{2}|t|dt \right)\otimes_{\C} \mathcal{Y}^{\ell}$. It follows from Sections \ref{iso}, and \ref{Imp} that in this realization

$$(\boldsymbol{H}_+)_{\ell}=L^2\left((0,\infty),\frac{1}{2}|t|dt \right)\otimes_{\C} \mathcal{Y}^{\ell},$$ and $$(\boldsymbol{H}_-)_{\ell}=L^2\left((-\infty,0),\frac{1}{2}|t|dt \right)\otimes_{\C} \mathcal{Y}^{\ell}.$$ 
We  denote the corresponding irreducible subrepresentations of $SL_2(\R)\times SO(3)$ by 
$(\pi_{\ell,+},(\boldsymbol{H}_+)_{\ell})$ and $(\pi_{\ell,-},(\boldsymbol{H}_-)_{\ell})$.

We recall that the Schwartz space $\mathcal{S}([0,\infty))$ of rapidly decreasing functions on $[0,\infty)$ is given by the space of all complex-valued functions on $[0,\infty)$ that are smooth on $(0,\infty)$, smooth from the right at $0$, and rapidly decaying at infinity; explicitly, for any $m,n\in \mathbb{N}_0$, $\sup\{|t^mf^{(n)}(t)|:t\in [0,\infty)\}< \infty$. 

Similarly, the  Schwartz space  $\mathcal{S}((-\infty,0])$ of rapidly decreasing functions on $(-\infty,0]$ is defined. 

\begin{proposition*}
For every $\ell\in \mathbb{N}_0$, 
\begin{enumerate}
    \item The space $\mathcal{S}(\boldsymbol{H}_{+})_{\ell}$ is isomorphic to the smooth subspace   of $(\boldsymbol{H}_+)_{\ell}$, where the latter space   is  given by 
    $t^{\ell}\mathcal{S}([0,\infty))\otimes_{\C} \mathcal{Y}^{\ell}$.
      \item The space $\mathcal{S}(\boldsymbol{H}_{-})_{\ell}$ is isomorphic to the smooth subspace    of $(\boldsymbol{H}_-)_{\ell}$, where the latter space    is $t^{\ell}\mathcal{S}((-\infty,0])\otimes_{\C} \mathcal{Y}^{\ell}$.
\end{enumerate}
\end{proposition*}

\begin{proof}
The isomorphism statement follows from the Casselman-Wallach  theorem as discussed in Subsection \ref{Imp}. 
    The rest of the proof immediately follows from the description of the smooth subspace of the discrete series representations of $SL_2(\R)$ in the Kirillov model as given in \cite{Baruch}.
\end{proof}

\subsection{Spectrum  of the Schr\"{o}dinger family is always real}\label{realspec}
Here we show that the  spectrum of the Schr\"odinger family in any admissible  smooth Fr\'echet representation of moderate growth is a subset of the real numbers. 
\begin{proposition*}
Let $(\pi,SL_2(\R),{W}) $ be an admissible  smooth Fr\'echet representation of moderate growth.  Then  for any $\kappa>0$
\[\operatorname{Spec}_{\boldsymbol{W}}\{S^{Ab}_{\kappa}\}\subseteq \R.\]
\end{proposition*}
\begin{proof}
Lemma \ref{333} shows that it is enough to prove the proposition for $\kappa=1$. We can further assume that $W$ is irreducible.  
By Casselman submodule theorem we can assume that $W$ is a quotient of a principal series representation realized on the space 
of $\nu$-homogeneous smooth functions on the punctured plane $\R^2_0:=\R^2\setminus\{0\}$ for some $\nu\in \C$.
Restriction to the circle $S^1\subset \R^2_0$ gives an isomorphic representation on the space $C^{\infty}(S^1)$ of smooth functions on the circle.

In this representation  the  Schr\"odinger family acts via 
\[\pi^{\nu}(S^{Ab}_{1}({\lambda}))=\left(1+i\nu (\lambda+1)s(\theta)  \right)+ic_{\lambda}(\theta)\frac{d}{d\theta},\]
where 
\[ s(\theta):=\sin(\theta)\cos(\theta),\quad  c_{\lambda}(\theta):=\cos^2(\theta)-\lambda \sin^2(\theta).  \]
The  Hermitian dual of $C^{\infty}(S^1)$ is
\[
\mathcal{D}^{\dagger}(S^1):=(C^{\infty}(S^1))^{\dagger}=\overline{\mathcal{D}'(S^1)}.
\]
The Hermitian adjoint operator 
\[
\pi^{\nu}(S^{Ab}_{1}({\lambda}))^{\dagger}:\mathcal{D}^{\dagger}(S^1) \longrightarrow \mathcal{D}^{\dagger}(S^1)
\]
is  given by 
\begin{eqnarray}\nonumber
&\pi^{\nu}(S^{Ab}_{1}({\lambda}))^{\dagger} =&   1-i(\overline{\nu}+2) (\overline{\lambda}+1)s(\theta)+ic_{\overline{\lambda}}(\theta) \frac{d}{d\theta}. 
\end{eqnarray}
For $\lambda \in \C\setminus \R$ (in fact this is true for $\lambda \in \C\setminus [0,\infty)$), the leading coefficient   $c_{\overline{\lambda}}(\theta)\neq 0$ never vanishes. Hence for $\lambda \in \C\setminus \R$,   any element in 
$\operatorname{Ker}(\pi^{\nu}(S^{Ab}_{1}({\lambda}))^{\dagger})$
must be represented by a smooth function on the circle which locally must be of the form 
\[C (c_{\overline{\lambda}}(\theta))^{-(\overline{\nu}+2)/2}\exp\left\{i\int_{\theta_0}^{\theta}\frac{dt}{c_{\overline{\lambda}}(t)}\right\},\]
for some complex constant  $C$ and some real constant $\theta_0$.
For existence of a global nonzero solution we must demand that the monodromy $M$, when  varying $\theta$ along  $[0,2\pi]$, is equal to $1$. 

The monodromy can be calculated explicitly, it is independent of $\nu$ and   is given by
\begin{eqnarray}\nonumber
&&  M=\exp\left\{i\int_0^{2\pi}\frac{d\theta}{c_{\overline{\lambda}}(\theta)} \right\}=\exp\left\{i\int_0^{2\pi}\frac{c_{\lambda}(\theta)d\theta}{|c_{\lambda}(\theta)|^2} \right\}. 
\end{eqnarray}
Hence 
\begin{eqnarray}\nonumber
&&  |M|=\exp\left\{-\int_0^{2\pi}\frac{\operatorname{Im} (c_{\lambda}(\theta))d\theta}{|c_{\lambda}(\theta)|^2} \right\}=\exp\left\{\operatorname{Im}(\lambda)\int_0^{2\pi}\frac{ \sin^2 (\theta)d\theta}{|c_{\lambda}(\theta)|^2} \right\}. 
\end{eqnarray}
Since we assume that  $\lambda\notin \R$ and since the integral is positive $|M|\neq 1$.
\end{proof}
We remark that similar but  more detailed calculations of the monodromy show that \[\operatorname{Spec}_{\boldsymbol{W}}\{S^{Ab}_{\kappa}\}\subseteq \left\{-\frac{\kappa^2}{n^2}
| n\in \N\right\}\cup[0,\infty).\]

\subsection{The self-adjoint property of the subspace of smooth vectors}\label{sv}
In this section we prove  Claim \ref{exa}.  

 Let $G$ be a real Lie group, $\fg$ the complexification of the Lie algebra $\operatorname{Lie}(G)$, and $\mathcal{U}(\fg)$ its universal enveloping algebra. 

Let $(\pi,G,\boldsymbol{H})$ be a unitary representation of $G$ on a Hilbert space $\boldsymbol{H}$. Denote by  $\boldsymbol{H}^{\infty}$ the subspace of $\boldsymbol{H}$ consisting of smooth vectors for $\pi$.  
\begin{claim*} 
    The natural action $\pi$ of the universal enveloping algebra $\mathcal{U}(\fg)$ on the space $\boldsymbol{H}^{\infty}$ is   self-adjoint.
\end{claim*}
Before proving the claim we  introduce relevant Sobolev spaces and other tools   that will be used in the proof.  

We denote by $\mathcal{U}^n(\fg)$   the $n$-th piece in the PBW filtration of  $\mathcal{U}(\fg)$ consisting of elements of  degree $\leq n$.

 We fix a $*$-invariant   Hermitian scalar product $B_n$ on  $\mathcal{U}^n(\fg)$ and   denote the corresponding Hermitian scalar product on $T^n:=\mathcal{U}^n(\fg)\otimes_{\C}\boldsymbol{H}$ by $\sigma_n$.  
 
 For $n\in \Z_{\geq 0}$, the $n$-th Sobolev norm of $v\in \boldsymbol{H}^{\infty}$ is defined as the operator norm of the morphism \begin{eqnarray}\nonumber
     && M_v:(\mathcal{U}^n(\fg), B_n)\longrightarrow (\boldsymbol{H}, \langle \_,\_ \rangle)\\ \nonumber
     && X\longmapsto \pi^{\infty}(X)v.
    \end{eqnarray}
The Fr\'echet topology on  $\boldsymbol{H}^{\infty}$ is induced by this sequence  of Sobolev norms.

 The adjoint action of $G$ on $\mathcal{U}^n(\fg)$ is continuous and we obtain a continuous representation of $G$ on    $T^n$.  The smooth subspace of $T^n$ is clearly given by $T^{n,\infty}=\mathcal{U}^n(\fg)\otimes_{\C}\boldsymbol{H}^{\infty}$. 

 The action $\pi$ defines a continuous linear map  $\pi^{n}:T^{n,\infty}\longrightarrow \boldsymbol{H}$ via
 \[X\otimes w\longmapsto  \pi(X)w. \]
 Hence any vector $v\in \boldsymbol{H}$ defines a linear functional $\xi_v$ on $T^{n,\infty}$ (in fact on all $\mathcal{U}(\fg)\otimes_{\C}\boldsymbol{H}^{\infty}$) via 
 \[\xi_v(X\otimes w)=\langle \pi(X)w,v \rangle. \]
\begin{defn*}
    The $n$-th  Sobolev space $\boldsymbol{H}^n\subset \boldsymbol{H}$ consists of all vectors $v\in \boldsymbol{H}$ such that the linear functional $\xi_v$ is bounded with respect to the  norm associated with $\sigma_n$. We denote the norm of the  functional $\xi_v$ by $\tau_n(v)$. The map $\tau_n:\boldsymbol{H}^n\longrightarrow \R_{\geq 0}$ is   a norm on $\boldsymbol{H}^n$.
\end{defn*}
In the proof of the claim we shall use the  following lemma   that can  be easily  verified.

\begin{lemma*}
\begin{enumerate}
    \item     The norm $\tau_n$ on $\boldsymbol{H}^n$ is induced by a Hermitian inner product $\langle\_,\_\rangle _n$. With respect to $\langle\_,\_\rangle _n$, $\boldsymbol{H}^n$ is a Hilbert space. 
    \item Each $\boldsymbol{H}^n$ is stable under $\pi(G)$, and this action of $G$ is continuous.
    \item $\boldsymbol{H}^{\infty}\subset \boldsymbol{H}^n$ and restriction of $\tau_n$ to 
$\boldsymbol{H}^{\infty}$ 
coincides with the   $n$-th Sobolev norm on   $\boldsymbol{H}^{\infty}$ that was introduced before.   
\end{enumerate}
\end{lemma*}

\begin{proof}[Proof of the claim]
From the definition of the adjoint action it is clear that  $(\boldsymbol{H}^{\infty})^{*}= \cap_{n}\boldsymbol{H}^n$.
 Note  that for $n>m$  the inclusion $\boldsymbol{H}^n\longrightarrow \boldsymbol{H}^m$  is a continuous  $G$-equivariant  map.  Hence it is enough to show that  $\boldsymbol{H}^{\infty}=\cap_{n}\boldsymbol{H}^n$. Clearly the inclusion $\boldsymbol{H}^{\infty}\subseteq \cap_{n}\boldsymbol{H}^n$ is a topological embedding. 
We will prove the required equality  using the fact that $\boldsymbol{H}^{\infty}$ is complete with respect to the system of Sobolev norms $\tau_n$. This follows from  \cite[Lemma. 2.15]{BK2014} (also see \cite{MR84713}).

 Let $v\in \cap_{n}\boldsymbol{H}^n\subset \boldsymbol{H}$.  
For any smooth measure $\mu$ with compact support  on $G$ we have $\pi(\mu)v\in \boldsymbol{H}^{\infty}\subseteq \cap_{k}\boldsymbol{H}^k$. 
 Using continuity of the actions of $G$ on the Sobolev spaces, 
  for any $n\in \Z_{\geq 0}$ we can find a smooth measure $\mu_n$ on $G$, supported in a small neighborhood of the identity in $G$, such that 
 $\tau_m(v-\pi(\mu_n)v)<2^{-n}$ for all $m\leq n$. 
 
The above inequality shows that the sequence of vectors $v_n:=\pi(\mu_n)v$ is a Cauchy sequence in $\boldsymbol{H}^{\infty}$ with respect to all norms $\tau_n$. Using completeness of  $\boldsymbol{H}^{\infty}$, this sequence has a limit in $\boldsymbol{H}^{\infty}$. The unique limit is clearly $v$, hence $v\in \boldsymbol{H}^{\infty}$.
\end{proof}

\printbibliography
 
\end{document}